\documentclass[journal,onecolumn]{IEEEtran}
\usepackage{latexsym,epsfig,subfigure,amssymb,dsfont}
\usepackage{amsmath}
\newtheorem{definition}{Definition}[section]
\newtheorem{remark}{Remark}[section]
\newtheorem{theorem}{Theorem}[section]

\newtheorem{lemma}{Lemma}[section]

\newtheorem{corollary}{Corollary}[section]

\newcommand{\eqdef} {\mbox{$\:\stackrel{\triangle}{=}\:$}}
\newcommand{\firstorder} {\mbox{$\:\stackrel{.}{=}\:$}}

\newcommand{\cX}{{\cal X}}
\def\bx{\mathbf{x}}
\def\bX{\mathbf{X}}
\def\bY{\mathbf{Y}}
\def\bQ{\mathbf{Q}}
\def\bZ{\mathbf{Z}}

\def\by{\mathbf{y}}
\def\bq{\mathbf{q}}

\def\b0{\mathbf{0}}

\def\bbZ{\mathbb{Z}}
\def\bbN{\mathbb{N}}

\newcommand{\fA}{{\mathfrak{A}}}
\newcommand{\cB}{{\mathcal{B}}}
\newcommand{\cG}{{\mathcal{G}}}
\newcommand{\cO}{{\mathcal{O}}}
\newcommand{\cS}{{\mathcal{S}}}
\newcommand{\cT}{{\mathcal{T}}}
\newcommand{\cC}{{\mathcal{C}}}
\newcommand{\tcS}{{\tilde{\mathcal{S}}}}
\newcommand{\tfA}{{\tilde{\mathfrak{A}}}}
\newcommand{\tcC}{{\tilde{\mathcal{C}}}}

\def\bz{\mathbf{z}}

\newcommand{\be}{\begin{equation}}
\newcommand{\ee}{\end{equation}}
\newcommand{\bea}{\begin{eqnarray}}
\newcommand{\eea}{\end{eqnarray}}
\newcommand{\bean}{\begin{eqnarray*}}
\newcommand{\eean}{\end{eqnarray*}}
\newcommand{\ben}{\begin{enumerate}}
\newcommand{\een}{\end{enumerate}}

\newcommand{\qed}{\hspace*{\fill}%
    \vbox{\hrule\hbox{\vrule\squarebox{.667em}\vrule}\hrule}\smallskip}
    \def\squarebox#1{\hbox to #1{\hfill\vbox to #1{\vfill}}}

\begin{document}

\bibliographystyle{ieeetr}
\baselineskip=1.90\normalbaselineskip

\title{Towards Exploring Fundamental Limits of System-Specific Cryptanalysis Within
Limited Attack Classes: Application to ABSG}

\author{Y\"{u}cel Altu\u{g},~\IEEEmembership{} M. K{\i}van\c{c} M{\i}h\c{c}ak~\IEEEmembership{Member}%
\thanks{The authors are with the Electrical and Electronic Engineering Department of Bo\u{g}azi\c{c}i University,
Istanbul, 34342, Turkey (e-mail: yucel.altug@boun.edu.tr, kivanc.mihcak@boun.edu.tr)}%
\thanks{Y.~Altu\u{g} is partially supported by T\"{U}B\.{I}TAK Career
Award no. 106E117; M.~K.~M{\i}h\c{c}ak is partially supported by
T\"{U}B\.{I}TAK Career Award no. 106E117 and T\"{U}BA-GEBIP Award.}}
\maketitle

\begin{abstract}
A new approach on cryptanalysis is proposed where the goal is to
explore the fundamental limits of a specific class of attacks
against a particular cryptosystem. As a first step, the approach is
applied on ABSG, which is an LFSR-based stream cipher where
irregular decimation techniques are utilized. Consequently, under 
some mild assumptions, which are common in cryptanalysis, the
tight lower bounds on the algorithmic complexity of successful
Query-Based Key-Recovery attacks are derived for two different
setups of practical interest. The proofs rely on the concept of
``typicality'' of information theory.
\end{abstract}

\section{Introduction}
\label{sec:intro}

In this paper, we introduce a (to the best of our knowledge) novel approach to cryptanalysis.
In our approach, the focus is jointly on a particular cryptosystem and a specific (sufficiently broad)
class of attacks of interest at the same time. Then, under some mild conditions, the goal is to
derive the {\em achievable fundamental  performance limit} for the attacks within the considered
class of interest against the cryptosystem at hand. The aforementioned limit should be ``achievable'',
in the sense that it is necessary to provide an explicit attack construction of which performance
coincides with the derived limit. Furthermore, the aforementioned limit should also necessarily be
``fundamental'',  in the sense that within the considered specific class, there does not
exist any attack of which performance is superior to the derived limit.

Our proposed approach contrasts with the trend in conventional
cryptanalysis, which can be outlined in two categories. In the first
category, the focus is on the  construction of  a generic attack,
which should be applicable (subject to slight modifications) to most
cryptosystems; common examples include time-memory tradeoff attacks
\cite{hellman80, bir-sha00}, correlation attacks \cite{sie85,
mei-sta89}, algebraic attacks \cite{cou-mei03, cou03} and alike. The
second category is conceptually on the opposite side of the
spectrum. Here, given a particular cryptosystem, the focus is on the
construction of a potentially-specialized attack, which is
``tailored'' specifically against the system at hand; hence, the
resulting attack is not applicable to a broader class of
cryptosystems in general. Although the approaches pursued in the
aforementioned two attack categories are radically different, it is
interesting to note that, for both of them the underlying
fundamental goal is the same, which can be summarized as providing a
``design advice'' to the cryptosystem designer. In practice, at first, the
cryptosystem designer is expected to test his/her proposed
system against generic attacks (first category); thus, such attacks
serve as a benchmark for the community of cryptosystem designers.
Next, the cryptanalyst tests a proposed cryptosystem via
constructing a cryptosystem-specific attack algorithm (second
category). Both categories have been shown to be extremely valuable in
practice since the first one provides a ``unified approach'' to
cryptanalysis via providing some generic attack algorithms and the
second one specifically tests the security of the considered
cryptosystem and consequently yields its potential weaknesses. On
the other hand, both categories of the conventional approach in
cryptanalysis lack to provide fundamental performance bounds, i.e.,
the question of ``what is the best that can be done?'' goes
unanswered. The main reason is that, for the first category, finding
out a fundamental performance bound necessarily requires considering
all possible cryptosystems, which is infeasible in practice; within
the second category, providing a fundamental performance bound
necessarily requires ``describing'' all possible cryptanalytic
propositions (in a computational sense) and quantifying the
resulting performances, which is again infeasible in practice.

In our proposed approach, we aim to derive  ``the best possible performance bound''
\footnote{Note that, this approach is analogous to providing both achievability and converse proofs
in classical information-theory problems. This connection will further be clarified throughout the paper.}
in a reasonably-confined setup.
Intuitively, we ``merge'' the first and the second categories of the conventional cryptanalytic approach;
we jointly focus on {\em both} a particular cryptosystem {\em and} a specific class of attacks, and subsequently
aim to analytically quantify the  fundamental, achievable performance bounds, i.e.,
specifically for a given cryptosystem,
our goal is to  find the achievable lower-bound on the complexity of a proposed class of attacks,
under a set of mild assumptions.
The main impact of this approach  is that, it aims to provide an advice for the cryptanalyst, instead of the
cryptosystem designer,  in contrast with the conventional approach.
If this resulting advice is ``positive'' (i.e., the fundamental achievable performance bound is of
polynomial complexity),  then the weakness of the analyzed cryptosystem is guaranteed
(which can also be achieved via pursuing the second category of the conventional cryptanalytic approach).
However, more interestingly, if the resulting advice is ``negative'' (i.e., the fundamental achievable
performance bound is of exponential complexity), then the considered class of attacks is {\em guaranteed}
to be useless, which, in turn, directs a cryptanalyst to consider different classes of attacks,
instead of experimenting with various attacks from the considered class via a (possibly educated) trial-and-error approach.
Thus, the negative advice case (for which this paper serves an exemplary purpose)
constitutes the fundamental value of our approach.
We believe that our efforts can be viewed as a contribution towards  the goal of
enhancing cryptanalytic approaches via incorporating a structural and procedural methodology.

In order to illustrate our approach, in this paper we consider a
class of Query-Based Key-Recovery attacks (of which precise
definition is given in Sec.~\ref{ssec:problem-formulation}) targeted
towards ABSG \cite{ABSG}, which is an LFSR(linear feedback shift
register)-based stream cipher that uses irregular decimation
techniques. Recall that, within the class of stream ciphers, the
usage of LFSR is an attractive choice due to the implementation
efficiency and favorable statistical properties of the LFSR output;
however, security of LFSR-based stream ciphers is contingent upon
applying additional non-linearities per the linear nature of LFSR
\cite{golomb}. An approach, which aims to achieve this task, is to
use irregular decimation techniques to the LFSR output \cite{ABSG,
BSG, shrinking, self_shrinking}. The motivation lying behind the
development of this approach is to render most conventional attacks
useless (such as algebraic attacks). Shrinking \cite{shrinking} and
self-shrinking generators (SSG) \cite{self_shrinking} are two
important examples of this approach. In particular, in the
literature SSG is well-known to be a very efficient algorithm and it
has been shown to possess favorable security properties
\cite{most_prob_SSG, zener_BDD_SSG, johansson_SSG}. The bit-search
generator (BSG) \cite{BSG} and its variant ABSG \cite{ABSG} are
newer algorithms, which also use irregular decimation techniques. In
\cite{gou-sib:06}, it has been shown that the efficiency (output
rate) of ABSG is superior to that of SSG and the security level of
ABSG is at least the same level provided by SSG under a broad class
of attacks. A detailed analysis of the statistical properties of
ABSG and BSG algorithms has recently been presented in \cite{Alt08}.
Since ABSG has been shown to be a state-of-the-art cryptosystem, in
our developments we focus on it under a reasonable class of attacks
and subsequently provide ``negative advices'' for the cryptanalyst
in various setups of interest. Next, we summarize our main results.

\textbf{Main Results:}
Our contributions, which have been derived under a set of mild assumptions
(specified in Sec.~\ref{ssec:assumptions}), are as follows:
\begin{itemize}

\item We show that breaking ABSG algorithm is equivalent to ``guessing''
a sequence of random variables, which are i.i.d. (independent
identically distributed) with geometric distribution of parameter
$1/2$ using complexity theoretic notions
(Theorem~\ref{theorem:attack-equivalence}).

\item In order to solve the problem mentioned in the previous item,
we formulate a sufficiently broad class of attacks, termed as
``Query-Based Key-Recovery attacks'', which are quite generic by construction,
and hence applicable for cryptanalysis for a wide range of cryptosystems (Definition~\ref{def:qubar}).

\item Within the class of attacks mentioned in the previous item,
first we concentrate on a practically-meaningful subset of them
(termed ``Exhaustive-Search Type Query-Based Key-Recovery attacks'')
(Sec.~\ref{sec:exhaustive-search}); we derive a fundamental lower
bound on the complexity of any successful attack in this subset
(Theorem~\ref{thrm:converse-exhaustive}); this lower bound is proven
to be  achievable to the first order in the exponent
(Theorem~\ref{thrm:achievability-exhaustive}).

\item We consider the set of all Query-Based Key-Recovery attacks (Sec.~\ref{sec:general-case});
we derive a fundamental lower bound on the complexity of any
successful attack within this set
(Theorem~\ref{thrm:converse-general}), followed by stating the proof
of the achievability result (to the first order in the exponent)
using the ``most probable choice'' attack given in \cite{ABSG}
(Theorem~\ref{thrm:achievability-general}).

\end{itemize}


\textbf{Organization of the Paper:} In
Section~\ref{sec:notation-background}, we present the notation used
in the paper and recall the definition of ABSG.
Section~\ref{sec:setup-formulation} provides the assumptions we have
employed throughout the paper, the problem formulation and the
definition of ``Query-Based Key-Recovery'' (QuBaR) attacks. In
Section~\ref{sec:exhaustive-search}, we derive a tight (to the
first order in the exponent) lower bound on the complexity of
exhaustive-search type QuBaR attacks. In
Section~\ref{sec:general-case}, we derive a tight lower bound on
the complexity of any QuBaR attack. We conclude with discussions
given in Section~\ref{sec:conclusion}.

\section{Notation and Background}
\label{sec:notation-background}

\subsection{Notation}
\label{ssec:notation}

Boldface letters denote vectors; regular letters with subscripts
denote individual elements of vectors. Furthermore, capital letters
represent random variables and lowercase letters denote individual
realizations of the corresponding random variable. The sequence of
$\left\{ a_1, a_2, \ldots , a_N \right\}$ is compactly represented
by $\mathbf{a}_1^N$. Given $x \in \{ 0,1 \}$, $\bar{x}$ denotes the
binary complement of $x$. The abbreviations ``i.i.d.'', ``p.m.f.''
and ``w.l.o.g.'' are shorthands for the terms ``independent
identically distributed'', ``probability mass function'' and
``without loss of generality'', respectively. 
Throughout the paper, all logarithms are base-$2$ unless otherwise specified.
Given a discrete random variable $X$ with the corresponding p.m.f. $p \left( \cdot
\right)$, defined on the alphabet $\cX$, its entropy (in bits) is $H(X) \eqdef
-\sum_{x \in \cX} p(x) \log p(x)$.
In the sequel, we say that ``$a_n$ and $b_n$ are equal to the first order in
the exponent'' provided that $\lim_{n \rightarrow \infty}
\frac{1}{n} \log \frac{a_n}{b_n} = 0$, which is denoted by $a_n
\firstorder b_n$ in our notation.

\subsection{Background}
\label{ssec:background}

Throughout this paper, we use the notation that was introduced in \cite{Alt08}.

\begin{definition}
Given an infinite length binary sequence $\bx = \left\{ x_n
\right\}_{n=1}^{\infty}$ which is an input to the ABSG algorithm, we define
\begin{itemize}
\item $\mathbf{\by} \triangleq\mathcal{A}\left(\bx\right)$, where the sequence
$\by$ represents the internal state of the ABSG algorithm and $y_i \in
\left\{ \varnothing , 0 , 1 \right\}$, $1 \leq i < \infty$. The
action of algorithm $\mathcal{A}$ is defined via the recursive mapping
$\mathcal{M}$:
\[
y_i=\mathcal{M}(y_{i-1},x_i), \quad 1 \leq i < \infty,
\]
with the initial condition $y_0 = \varnothing$. The mapping
$\mathcal{M}$ is defined in Table~\ref{tab:mapping} .

\begin{table}[h]\caption{Transition Table of algorithm $\mathcal{A}$}
\begin{center}
\begin{tabular}{c|cc}
  $y_{i-1} \backslash x_i$ & 0 & 1 \\
  \hline
  $\varnothing$ & 0 & 1 \\
  0 & $\varnothing$ & 0 \\
  1 & 1 & $\varnothing$ \\
\end{tabular}
\end{center}
\label{tab:mapping}
\end{table}

\item $\bz \triangleq \cB\left(\by\right)$, where the sequence
$\bz$ represents the output of the ABSG algorithm, such that the
action of the algorithm $\cB$ is given as follows:
\begin{equation*}
z_j =\left\{ \begin{array}{cl} y_{i-1}, & \textrm{ if } \, y_i=\varnothing \textrm{ and } y_{i-2}=\varnothing,\\
\bar{y}_{i-1}, & \textrm{ if } \, y_i=\varnothing \textrm{ and } y_{i-2} \neq \varnothing,\\
\end{array} \right.
\end{equation*} where $j \leq i$ and $i,j \in \mathbb{Z}^{+}$.
\end{itemize}
\label{def:ABSG}
\end{definition}

From Definition~\ref{def:ABSG}, we clearly deduce that the ABSG algorithm produces an output
bit ($z_j$ denoting the $j$-th output bit) if and only if the value of the corresponding internal state variable ($y_i$
denoting the value of the internal state variable at time $i$) is $\varnothing$. The fact that $y_i \neq \varnothing$ for all $i$
 is the reason of the mismatch
between the input sequence indices (which are the same as the indices of the internal state variables) and the output
sequence indices.

\section{Problem Setup and Formulation}
\label{sec:setup-formulation}

\subsection{Assumptions and Preliminaries}
\label{ssec:assumptions}

Throughout this paper, we consider the type of attacks, in which retrieving
$L$ (where $L$ is the degree of the feedback polynomial of the
generating LFSR) linear equations in terms of $\bx_1^M$ is aimed.
This type of attacks correspond to {\em key recovery attacks} to ABSG
(assuming that the feedback polynomial of LFSR is known to the
attacker, which is a common assumption in cryptanalysis). In particular, within the
class of key recovery attacks, we concentrate on {\em query-based key recovery
attacks} (abbreviated as {\bf {\em ``QuBaR attacks''}}  in the rest of the paper); QuBaR
attacks shall be defined formally in Sec.~\ref{ssec:problem-formulation}. The following
assumptions are made in this attack model:

\begin{itemize}
\item[\bf{A1}:] The length-$M$ input sequence $\bx_1^M$ is
assumed to be a realization of an i.i.d. Bernoulli process with
parameter $1/2$.

\item[\bf{A2}:] The length-$N$ output sequence $\bz_1^N$ is assumed to be
given to the attacker, where $N, M \in \bbZ^+$ (note that, this implies
we necessarily have $M>N \geq 1$ due to Definiton~\ref{def:ABSG}).

\item[\bf{A3}:] Explicit knowledge of the feedback
polynomial of the generating LFSR is not used.

\item[\bf{A4}:] The degree of the feedback polynomial of the generating LFSR,
i.e., $L$, is sufficiently large.

\end{itemize}

Note that assumption A3 will be further clarified after we describe
QuBaR attack model precisely. Further, from now on we denote the
input sequence as $\bX_1^M$ and the corresponding internal state
sequence as $\bY_1^M$ due to the stochastic nature of the input and
hence the internal state sequences. Next, we continue with the
following definitions.

\begin{definition}
The symbol $H_i$ denotes the index of the $i$-th $\varnothing$ in $\bY_1^M$, for $0 \leq i \leq N$.
\label{def:h}
\end{definition}

Note that, since we have $Y_0 = \varnothing$ with probability $1$ by convention, we also use
$H_0 = 0$ with probability $1$ as the initial condition for $\left\{ H_i \right\}$.

\begin{definition}
We define $Q_i \eqdef H_{i} - H_{i-1} - 2$,  for $1 \leq i \leq N$.
\label{def:Q}
\end{definition}

\begin{remark}
For each $Q_i$ (regardless of its particular realization), the ABSG algorithm generates
an output bit $z_i$. Thus, the number of output bits in the ABSG algorithm is precisely equal to the
number of corresponding $\left\{ Q_i \right\}$.
\label{rem:Q-outputbit}
\end{remark}
Next, we state the following result regarding the distribution of $\left\{ Q_i \right\}$, which
will be heavily used throughout the rest of the paper.

\begin{lemma}
Under assumptions A1 and A2, the random variables $\left\{ Q_i
\right\}$ are i.i.d. with geometric p.m.f. of parameter $1/2$:
\begin{equation}
p \left( q_i \right) \eqdef \Pr \left[ Q_i = q_i | \bz_1^N \right] = \left( 1/2 \right)^{q_i+1}, \textrm{
for } q_i \in \bbN, \; 1 \leq i \leq N. \label{eq:Q-pmf}
\end{equation}
\label{lem:Q-pmf}
\end{lemma}

\begin{proof}
See Appendix~\ref{app-1}.
\end{proof}

\subsection{Problem Formulation}
\label{ssec:problem-formulation}

In this section, we provide an analytical formulation of the problem considered in this paper.
As the first step, we show that, under assumptions A1, A2, A3, and A4, all key recovery attacks to
ABSG are equivalent to recovering the exact realizations of $\bQ_1^N$,
stated in Theorem~\ref{theorem:attack-equivalence}\footnote{For the random variable $Q_i$,
its realization is denoted by $q_i$.}:

\begin{theorem}
Under the assumptions A1, A2, A3 and A4, the following three computational problems are equivalent
in the sense of probabilistic polynomial time reducibility \cite{Goldreich01}:
\begin{enumerate}
\item Retrieving any $L$ independent linear equations in terms of $\bX_1^M$.
\item Retrieving any $L$ consecutive bits from $\bX_1^M$.
\item Correctly guessing $\bQ_i^{\theta + i -1}$ for any positive integers $i$ and $\theta$ such that
\begin{equation}
\sum_{j=i}^{\theta+i-1} \left( q_j + 2 \right) \geq L ,
\label{eq:pbm-equivalence-constraint}
\end{equation}
is satisfied.
\end{enumerate}
\label{theorem:attack-equivalence}
\end{theorem}
\begin{proof}
See Appendix~\ref{app-2}.
\end{proof}

Next, we introduce the model for the query type attacks, namely {\em
QuBaR attacks}, which are considered throughout the paper.
Qualitatively, a QuBaR attack consists of repeating the following
procedure: For a cryptosystem that has a secret, generate a
``guess'', which aims to guess the secret itself, and subsequently
``checks'' whether the guess is equal to the secret or not; if the
guess is equal to the secret,  then terminate the procedure, else
continue with another guess. The maximum number of guesses proposed
in this procedure are limited by the complexity of the QuBaR attack,
which is provided as an input parameter to the attack algorithm.
Note that, if the task at hand is to guess i.i.d. random variables
(which is the case for the third problem of
Theorem~\ref{theorem:attack-equivalence}), the QuBaR attack model is
intuitively obviously reasonable. Furthermore, recall that most of
the cryptanalysis against symmetric key cryptography may be modeled
in this way (e.g., time-memory attacks, correlation attacks,
algebraic attacks and alike). Next, we formally present the general
form of QuBaR attack algorithms.

\begin{definition}
Assuming the existence of a ``check algorithm'' $\cT \left( G \right)$ for a ``guess''
$G$ (the output of $\cT \left( G \right)$  is $1$ if and only if the guess $G$ is equal to the secret),
a QuBaR attack  algorithm, of complexity $\mathcal{C}$, executes the following steps: \\
\begin{center}
\begin{tabular}{|p{10cm}|}
\hline
\\
\vspace{-0.5cm}
For $k=1$ to $\mathcal{C}$ \\
{\bf 1.} Generate a guess $G_k$. \\
{\bf 2.} Compute $\cT(G_k)$. \\
{\bf 3.} If $\cT(G_k) = 1$, then terminate and output the secret
given by $G_k$.\\
\vspace{0.05cm} end
\\
\hline
\end{tabular}
\end{center}
\label{def:qubar}
\end{definition}
$ $

\indent Next, we introduce the particular ``guess'' structure
(together with the accompanying relevant definitions) which aims to
find $\bQ_i^{\theta+i-1}$ so as to solve the third computational
problem of Theorem~\ref{theorem:attack-equivalence}.
\begin{definition}
An \emph{ABSG-guess} is a triplet defined as $G \eqdef \{i, \theta,
\bq_{i}^{\theta+i-1}\}$, such that $2\theta + \beta \geq L$,  where
$\beta \eqdef \sum_{j=0}^{\theta-1}q_{i+j}$, $i \geq 1$ and $i +
\theta -1 \leq N$. \label{def:guess}
\end{definition}

The Bernoulli random variable, $\cT \left( G_k \right)$, indicates the success probability
of guess $G_k$ and is heavily used
throughout the rest of the paper, where $G_k \eqdef \left( i_k ,
\theta_k , \bq_{i_k}^{\theta_k + i_k -1} \right)$ is the ABSG-guess
of a QuBaR attack (against ABSG) at step $k$. Note that, at each
step $k$, the ``guessed'' values $\bq_{i_k}^{\theta_k + i_k -1}$
themselves depend on $k$, which is not explicitly stated (unless otherwise specified) for the
sake of notational convenience; this should be self-understood from
the context.

\begin{remark}
Note that, the probability of having a successful QuBaR attack {\em after
precisely $K$ steps} is equal to $\Pr \left[ \cT\left( G_1 \right)=0, \, \right.$
$\left. \cT\left( G_2 \right)=0 , \,
\ldots \,  \cT\left( G_{K-1} \right)=0 , \, \cT\left( G_K \right)=1
\right]$ which is {\em not} equal to $\Pr \left( \cT\left( G_K
\right)=1 \right)$ (the latter being equal to the marginal successful guess
probability at step $K$).
Moreover, neither of these expressions is the success probability of any QuBaR
attack with a specified complexity, which will formally be defined in (\ref{eq:prob-succ}).
Observe that our formulation
allows the usage of {\em potentially correlated guesses} $\left\{
G_k \right\}$ which aims to make the approach as generic as
possible. \label{rem:I-k-def}
\end{remark}

\begin{corollary}
Per Lemma~\ref{lem:Q-pmf} and Definition~\ref{def:guess}, we have
\begin{equation}
\Pr \left[ \cT\left( G_k \right)= 1 \right] = \Pr \left[
\bQ_{i_k}^{i_k + \theta_k - 1}  = \bq_{i_k}^{i_k + \theta_k - 1} \,
| \, \bz_1^N \right] = \prod_{j=i_k}^{i_k + \theta_k -1} \left(
\frac{1}{2} \right)^{q_j+1} =  \left( \frac{1}{2} \right)^{\beta_k +
\theta_k} , \label{eq:pmf-I}
\end{equation}
where $\beta_k \eqdef \sum_{j=0}^{\theta_k-1} q_{i_k+j}$.
\label{cor:pmf-I}
\end{corollary}

The following corollary, which is a direct consequence of
Theorem~\ref{theorem:attack-equivalence}, is one of the key results
of the paper.
\begin{corollary}
All QuBaR-type attacks against ABSG are probabilistic polynomial
time reducible to the QuBaR algorithm (defined in
Definition~\ref{def:qubar}) which uses ABSG-guesses defined in
Definition~\ref{def:guess} and aims to find $\bQ_i^{\theta+i-1}$
satisfying (\ref{eq:pbm-equivalence-constraint}) for any $i , \theta
\in \bbZ^+$. \label{cor:attack-model}
\end{corollary}

\begin{definition}
From now on, we call an arbitrary ``ABSG-Guess'', $G$, simply as ``guess''.
Further, for the sake of notational convenience, we use
\[
\mathfrak{A} = \left\{ G_k \right\}_{k=1}^{\cC \left( \fA \right)}
\]
for any attack algorithm $\fA$ mentioned in Corollary~\ref{cor:attack-model},
where $\cC \left( \fA \right)$ denotes the (algorithmic) complexity of $\fA$ (i.e., number
of guesses applied within $\fA$).
Accordingly, the success probability of any $\fA$ is given by
\begin{equation}
\mbox{Pr}_{succ} \left( \fA \right) \eqdef \Pr \left[
\vee_{k=1}^{\cC \left( \fA \right)} \cT\left( G_k \right)=1 \right]
= 1 - \Pr\left[\wedge_{k=1}^{\cC \left( \fA \right)}\cT\left( G_k
\right)=0\right] . \label{eq:prob-succ}
\end{equation}
\label{def:attack-algorithm}
\end{definition}

\vspace{-0.3in} Hence, as far as QuBaR attacks against ABSG are
concerned, w.l.o.g., in this paper we focus on the ones specified in
Corollary~\ref{cor:attack-model}, which aim to solve the third
computational problem of Theorem~\ref{theorem:attack-equivalence}.
In particular, in the rest of the paper, we explore the fundamental
limits of the aforementioned QuBaR attacks (denoted by $\fA$) under
various setups of interest.

\vspace{1in}
\begin{remark} $ $
\begin{itemize}
\item[(i)]  \underline{Measure of QuBaR Complexity in Terms of $L$}:
At first glance, it may look reasonable to evaluate the
complexity of a QuBaR attack in terms of the length of its input,
which is $N$ since the input is $\bz_1^N$. Note that, this is a
common practice in complexity theory. However, when we confine the
setup as the application of a QuBaR attack to the ABSG algorithm
(prior to which there exists an LFSR whose length-$L$ initial state
is unknown), then it would be more reasonable to evaluate the
complexity of a QuBaR attack in terms of $L$  (since we
eventually aim to find $L$ consecutive bits of $\bX_1^M$; see
Theorem~\ref{theorem:attack-equivalence}). This is precisely the
approach we pursue in this paper, i.e., the analysis of the
resulting QuBaR attack complexity is given as a function of $L$.

\item[(ii)] \underline{Time Complexity of QuBaR}:
First, note that the time complexity of a  QuBaR attack (denoted by $\fA$) is given 
by the product of $\cC \left( \fA \right)$, the complexity of generating a guess
and the complexity of checking a guess. Hence, the quantity $\cC \left( \fA \right)$
forms a lower bound on the time complexity of the QuBaR attack, $\fA$. 
Furthermore, complexities of both generating a guess and checking a guess
may, in practice,  be considered to be of $poly \left( L \right)$ (see item (v) of 
this remark). Moreover, we will soon show that at optimality $\cC \left( \fA \right)$ 
is of $exp \left( L \right)$. Hence, at optimality, the lower bound of $\cC \left( \fA \right)$ 
is, in practice, tight to the first order in the exponent. 
Therefore, throughout this paper, we ``treat'' the quantity of $\cC \left( \fA \right)$ 
as the time complexity of a QuBaR attack $\fA$ and carry out the analysis accordingly.

\item[(iii)] \underline{Data Complexity of QuBaR}:
First, note that, for a QuBaR attack $\fA$, consisting of guesses
$\left\{ G_k \right\}$,  the data complexity of the $k$-th guess 
$G_k = \left( i_k , \theta_k , \bq_{i_k}^{\theta_k+i_k-1} \right)$ is, 
by definition, $\theta_k$. We will soon show that, at ``general case'' optimality 
we have $\theta_k = \cO \left( L \right)$ for each $k$. Hence, 
the data complexity of an optimal QuBaR attack $\fA$ is at most 
$\cC \left( \fA \right) \cdot \cO \left( L \right)$. Furthermore, we will show that
at optimality $\cC \left( \fA \right)$ is of $exp \left( L \right)$. Hence, 
we conclude that, at optimality  $\cC \left( \fA \right)$ is a tight (to the first order
in the exponent) upper bound on the data complexity\footnote{This result is valid for the 
general case QuBaR attacks, analyzed in Sec.~\ref{sec:general-case}. For a restricted class of
QuBaR attacks, namely ``Exhaustive-Search Type'' QuBaR attacks (analyzed in 
Sec.~\ref{sec:exhaustive-search}),  we show that, 
at optimality $\cC \left( \fA \right)$ is a loose upper bound on the data complexity.}.  
 
\item[(iv)] \underline{Algorithmic Complexity of QuBaR}:
In parts (ii) and (iii) above, we stress that for an optimal QuBaR
attack algorithm $\fA$ against ABSG, $\cC \left( \fA \right)$ forms
a {\em tight} lower (resp. upper) bound on the time (resp. data) complexity of
$\fA$, to the first order in the exponent\footnote{Once again, the argument in this remark is valid
for the general case QuBaR attacks of Sec.~\ref{sec:general-case}.}. 
Following the general convention in cryptanalysis, we use the
term ``algorithmic complexity'' as the maximum of time complexity
and data complexity. Thus, 
we conclude that,  at optimality the algorithmic complexity is equal to the time complexity,  
Furthermore, at optimality, the time complexity, the data complexity and $\cC \left( \fA \right)$
are all equal to each other to the first order in the exponent. 
Our subsequent developments are based on analytical quantification of $\cC \left( \fA \right)$. 
Moreover, due to the aforementioned reasons, our results on $\cC \left( \fA \right)$ apply (to the 
first order in the exponent)  to the time complexity, the data complexity and the algorithmic complexity, as well.

\item[(v)] \underline{Practical Implementation Approaches to QuBaR Algorithms}:
As far as practical attacks are concerned, existence of a
polynomial-time guess generation algorithm is obvious. Furthermore,
a polynomial-time check algorithm, which corresponds to the
procedure of initiating a LFSR (whose feedback polynomial is assumed
to be known) with the corresponding ``guessed and retrieved'' $L$
consecutive bits of $\bX_1^M$,  generating sufficiently many output
bits and comparing them with the original output bits, constitutes a
practical approach.

\item[(vi)] \underline{Relationship Of QuBaR Attacks With State-Of-The-Art
Attack Algorithms}: We see that QuBaR attacks are analogous to ``first type of
attacks'' described in \cite{gou-sib:06}, which ``aim to exploit
possible weaknesses of compression component introduced by ABSG''.
However, note that, QuBaR attacks {\em do not} use explicit knowledge of
the feedback polynomial of the generating LFSR, (recall the structure of
algorithm $\cT$) which is a direct consequence of the assumption A3.

\end{itemize}
\label{rem:algorithm-description}
\end{remark}

\section{Optimum Exhaustive-Search Type QuBaR Attacks Against ABSG}
\label{sec:exhaustive-search}

In this section, we deal with ``exhaustive-search'' type QuBaR
attacks which are formally defined in
Definition~\ref{def:exhaustive-search}. Qualitatively, given the
output sequence $\bz_1^N$, an exhaustive-search type QuBaR attack
aims to correctly identify $\theta$-many $\left\{Q_i \right\}$
(equivalently at least $L$ consecutive bits of $\bX_1^M$ per
Theorem~\ref{theorem:attack-equivalence}) beginning from {\em an
arbitrarily-chosen, fixed index}, subject to constraint
(\ref{eq:pbm-equivalence-constraint}) \footnote{In contrast with
exhaustive-search attacks, we also consider a generalized version,
where we focus on identifying $\theta$-many $\left\{Q_i \right\}$,
possibly beginning from arbitrarily-chosen, multiple indices, which
constitutes the topic of Sec.~\ref{sec:general-case}.}. Since the
attacker is confined to initiate the guesses beginning from a fixed
index for exhaustive-search attacks, in practice this can be thought
to be equivalent to a scenario where the attacker uses only a {\em
single portion} of the observed output sequence $\bz_1^N$.

First theorem of this section, namely
Theorem~\ref{thrm:achievability-exhaustive}, proves the existence of
an exhaustive-search type QuBaR attack with success probability of
$1 - \epsilon$ (for any $\epsilon > 0$) with algorithmic complexity
$2^{2 L / 3}$ (in particular, with time complexity $2^{2L/3}$ and
data complexity $L/3$) under the assumptions mentioned in
Section~\ref{ssec:assumptions}. The second theorem of this section,
namely Theorem~\ref{thrm:converse-exhaustive}, proves that the
algorithmic complexity of the best (in the sense of $\mathcal{C}$)
exhaustive-search type QuBaR algorithm under the assumptions A1, A2,
A3, A4 is lower-bounded by $2^{2 L / 3}$ (to the first order in the
exponent). Hence, as a result of these two theorems, we show that the
overall algorithmic complexity of the best exhaustive-search attack
against ABSG has complexity $2^{2L/3}$ to the first order in the
exponent (argued in Corollary~\ref{cor:exhaustive-lower-bound}).
Note that, in \cite{gou-sib:06} Gouget et. al. mention the existence
of an exhaustive-search attack (under i.i.d. Bernoulli $1/2$ input
assumption) of complexity $\mathcal{O} \left( 2^{2 L /3} \right)$
without providing the details of the attack. Our main novelty in
this section is that, we provide a rigorous proof about the
existence of such an attack
(Theorem~\ref{thrm:achievability-exhaustive}, which is analogous to
the ``achievability''-type proofs in traditional lossless source
coding) and further show that this is the best (to the first order
in the exponent) in the sense of algorithmic complexity under some
certain assumptions, specifically within the class of
exhaustive-search QuBaR attacks
(Theorem~\ref{thrm:converse-exhaustive}, which is analogous to the
``converse''-type proofs in traditional lossless source coding). As
a result, the developments in this section can be considered to be
analogous to those of source coding by Shannon \cite{sha:48}; see
Remark~\ref{rem:neg-result-exhaustive} for a further  discussion on
this subject. Theorem~\ref{thrm:optimality-exhaustive} concludes the
section, which characterizes some necessary conditions of the
optimal exhaustive-search type QuBaR attacks against ABSG.

We begin our developments with the formal definition of exhaustive-search
type QuBaR attacks.

\begin{definition}
The class of exhaustive-search type QuBaR attacks against ABSG are
defined as
\begin{equation}
\mathcal{S}^{E} \eqdef \{ \fA^E = \{ G_k\}_{k=1}^{\cC \left( \fA^E
\right) } : \forall k, i_k=1\}, \label{eq:exhaustive-search}
\end{equation}
where each $k$-th guess $G_k = \left( i_k , \theta_k ,
\bq_{i_k}^{\theta_k + i_k -1} \right)$ is subject to
(\ref{eq:pbm-equivalence-constraint}) (see
Definition~\ref{def:guess}). \label{def:exhaustive-search}
\end{definition}

\begin{remark}
Exhaustive-search type attacks constitute an important class of
attacks in cryptanalysis. They essentially determine the ``effective
size'' of the key space of any cipher. In case of ABSG,
as we mentioned at the beginning of this section, since the exhaustive-search
type QuBaR attack uses a single portion of the output sequence,
they form a basic choice for practical cryptanalysis via QuBaR
attacks in situations where a limited amount ($poly \left( L \right)$) of
output data are available to the attacker.
\label{rem:sign-exhaus-search}
\end{remark}

Thus, at each $k$-th step, via guess $G_k$ an exhaustive-search type
QuBaR attack aims to correctly identify $\theta_k$-many $\left\{ Q_i
\right\}$ subject to (\ref{eq:pbm-equivalence-constraint}) beginning
from a fixed index $i_k$, equivalently at least $L$ consecutive bits
of $\bX_1^M$ beginning from the index $i'_k$ (in general $i'_k \neq
i_k$ due to the ``decimation'' nature of ABSG). As we specified in
Definition~\ref{def:exhaustive-search}, in our developments w.l.o.g.
we use $i_k = 1$ (which in turn implies having $i'_k=1$ as well).

\begin{theorem}
{\em (Achievability - Exhaustive-Search)} Under the assumptions A1,
A2, A3, A4, mentioned in Section~\ref{ssec:assumptions}, there
exists an exhaustive-search type QuBaR attack algorithm
$\fA_{ach,opt}^E$ against ABSG with $\cC \left( \fA_{ach,opt}^E
\right)
 = 2^{ 2 L / 3}$ such that $\Pr_{succ} \left( \fA_{ach,opt}^E \right) >
 1 - \epsilon$, for any $\epsilon > 0$.
 Further, $\cC_{ave} \left( \fA_{ach,opt}^E \right) =  \frac{1}{2} \left( 2^{2L/3}+1\right)$
where $\cC_{ave} \left( \fA_{ach,opt}^E \right)$ is the {\em
expected complexity} of $\fA_{ach,opt}^E$ over the probability
distribution induced by $\bq$. \label{thrm:achievability-exhaustive}
\end{theorem}

\begin{proof}
See Appendix~\ref{app-3}.
\end{proof}

\begin{remark}
An inspection of the proof of
Theorem~\ref{thrm:achievability-exhaustive} reveals that (as
promised in Remark~\ref{rem:algorithm-description}) the overall data
complexity of the proposed attack algorithm $\fA_{ach,opt}^E$ is
$L/3$ which certainly implies that each guess is of data complexity
$\cO \left( L \right)$. Furthermore, the overall time complexity of
$\fA_{ach,opt}^E$ is $\cO \left( 2^{2 L / 3} \right)$ assuming that
the contribution of the generation of each guess is $poly \left( L
\right)$ (which is reasonable in practice). Note that, the time and
data complexity of the proposed attack $\fA_{ach,opt}^E$  used in
the proof of Theorem~\ref{thrm:achievability-exhaustive} coincides
with the one mentioned in \cite{gou-sib:06}.
\label{rem:achievability-remark}
\end{remark}

Next, we prove the converse counterpart of Theorem~\ref{thrm:achievability-exhaustive}, namely derive a
lower bound on the algorithmic complexity of any exhaustive-search type QuBaR attack
with an inequality constraint on the success probability.

\begin{theorem}
{\em (Converse - Exhaustive-Search)} Under the assumptions A1, A2,
A3, A4,  and for any $\fA^E \in \mathcal{S}^{E}$ with
$\Pr_{succ}\left( \fA^E \right) > \frac{1}{2}$, we necessarily have
$\cC \left( \fA^E \right)
> \underline{\cC}_{min}^E \eqdef 2^{2L/3}\left( \frac{1}{2} - \frac{6}{L}\right)$.
\label{thrm:converse-exhaustive}
\end{theorem}

\begin{proof}
See Appendix~\ref{app-4}.
\end{proof}

\begin{corollary}
After some straightforward algebra, it can be shown that
\[
\cC \left( \fA_{ach,opt}^E \right) \firstorder \cC_{ave} \left(
\fA_{ach,opt}^E \right) \firstorder \underline{\cC}_{min}^E
\]
in $L$. Thus, Theorems~\ref{thrm:achievability-exhaustive} and
\ref{thrm:converse-exhaustive} show that, under the assumptions
mentioned in Section~\ref{ssec:assumptions}, the {\em tight} lower
bound (to the first order in the exponent) on the algorithmic
complexity of any exhaustive-search type QuBaR attack against ABSG
is $2^{2 L/3}$. \label{cor:exhaustive-lower-bound}
\end{corollary}

Following remark provides the promised discussion at the beginning
of the section, which interprets the relationship between the result
proved in this section (namely,
Theorems~\ref{thrm:achievability-exhaustive} and
\ref{thrm:converse-exhaustive}) and the traditional lossless source
coding of information theory.

\begin{remark}
Observe that for the exhaustive-search setup, the problem is
``somewhat dual'' of the lossless source coding problem.
Intuitively, the concept of cryptographic compression (which is also
termed as ``decimation'' in this paper) aims to produce a sequence
of random variables, such that the sequence is as long as possible
with the highest entropy possible so as to render cryptographic
attacks useless as much as possible (which amounts to making the
decimation operation ``non-invertible'' in practice). On the other
hand, in lossless source coding, the goal is to produce an output
sequence which is as short as possible while maintaining ``exact
invertibility'' (which amounts to ``lossless'' decoding). Hence, it
is not surprising that, from the cryptanalyst's point of view, usage
of concepts from lossless source coding may be valuable. To be more
precise, the cryptanalyst aims to identify a set of {\em
highly-probable} sequences (each of which is a collection of i.i.d.
random variables from a known distribution), of which cardinality is
as small as possible, thereby maximizing the chances of a successful
guess with the least number of trials. As a result, the usage of the
concept of {\em typicality} fits naturally within this framework. In
particular, typicality is the essence of the proof of the converse
theorem (Theorem~\ref{thrm:converse-exhaustive}), which states a
fundamental lower bound on the complexity of all possible
exhaustive-search type QuBaR attacks. {\em The outcome of
``converse'' states a negative result (which is unknown for the case
of stream ciphers to the best of our knowledge) within a reasonable
attack class in cryptanalysis by construction}. This observation
contributes to a significant portion of our long-term goal, which
includes construction of a unified approach to cryptanalysis of
stream ciphers. In particular, our future research includes focusing
on specific cryptosystems and quantifying fundamental bounds on the
performance of attacks (within a pre-specified reasonable class)
against these systems. \label{rem:neg-result-exhaustive}
\end{remark}

Following theorem characterizes some important necessary conditions
for an optimal exhaustive-search type QuBaR attack against ABSG,
subject to an equality constraint on the success probability.
Thus, these results are important in practice since they provide some guidelines
in construction of optimal or near-optimal exhaustive-search type QuBaR attacks.

\begin{theorem}
Given an optimal (in the sense of minimizing $\mathcal{C}\left(
\fA_{opt}^E \right)$ subject to an equality constraint on the success
probability) exhaustive-search type QuBaR attack (denoted by
$\fA_{opt}^E$) against ABSG, we have the following necessary conditions:
\begin{itemize}
\item[(i)] The corresponding guesses are {\em prefix-free}.

\item[(ii)] The corresponding ``success events'' $\left\{ \cT\left( G_i \right) = 1 \right\}_{i=1}^{\cC\left( \fA_{opt}^E\right)}$ are
{\em disjoint}.

\item[(iii)] We have
\begin{equation}
\mbox{Pr}_{succ}\left( \fA_{opt}^E \right) = \Pr\left(
\vee_{k=1}^{\mathcal{C}\left(\fA_{opt}^E\right)} \left[ \cT\left(
G_k\right) = 1 \right] \right) = \sum_{k=1}^{\mathcal{C}\left( \fA_{opt}^E
\right)} \Pr\left( \cT\left( G_k\right) = 1\right).
\label{eq:practicalcode-unionboundoptimality}
\end{equation}

\item[(iv)] The corresponding
``success events'' $\left\{ \cT\left(G_i\right) = 1
\right\}_{i=1}^{\cC\left( \fA_{opt}^E\right)}$ satisfy
\[
 \left( i > j \right) \quad \Longrightarrow \quad
\left[ \Pr\left( \cT\left(G_i\right)=1 \right) \leq \Pr\left(
\cT\left(G_j\right)=1 \right) \right],
\]
for any $i \neq j$, such that, $i,j \in \left\{1, \ldots , \cC\left(
\fA_{opt}^E \right) \right\}$.

\end{itemize}
\label{thrm:optimality-exhaustive}
\end{theorem}

\begin{proof}
See Appendix~\ref{app-5}.
\end{proof}

\section{Optimum QuBaR Attacks Against ABSG (General Case)}
\label{sec:general-case}

In this section, we consider the ``general case QuBaR attacks'', i.e.,
we relax the condition of being ``exhaustive-search'', which amounts to
relaxing the condition of $i_k = 1$ for $\left\{ G_k \right\}$ in \eqref{eq:exhaustive-search}.
Thus, the goal of the attacker is to guess the true values of $\bQ_i^{i-\theta+1}$, subject to
\eqref{eq:pbm-equivalence-constraint}, for an arbitrary initial index  $i$,
equivalently (cf. Theorem~\ref{theorem:attack-equivalence}) the attacker's
goal is to retrieve \emph{any} (at least) $L$ consecutive bits from the input
sequence $\bX_1^M$. Note that, this setup implies that exponential amount
of output bits are available to the attacker for the cryptanalysis. As we will
show in the sequel, via following this formulation, we can improve the
time-complexity (and hence the overall algorithmic complexity) at the expense
of an exponential increase in the data complexity (which does not affect the
overall algorithmic complexity). Thus, the general case can be viewed as
one extreme regarding the time-data tradeoff; the other extreme is the exhaustive-search
type attacks covered in the previous section.

Similar to the exhaustive-search case, we prove an achievability
result first, namely Theorem~\ref{thrm:achievability-general},
(which is simply the ``most-probable choice attack'' of \cite{ABSG,
gou-sib:06}), which implies the existence of a QuBaR attack of
algorithmic complexity $2^{L/2}$ under the assumptions mentioned in
Section~\ref{ssec:assumptions}. Next, we provide the converse
theorem for the general case, which states that  the best QuBaR
attack's algorithmic complexity is lower bounded with $2^{L/2-1}$
under the assumptions A1, A2, A3, A4. Hence, we conclude that, to
the first order in the exponent, the best QuBaR attack against ABSG
is of complexity $2^{L/2}$.

We begin our development with the following definition.
\begin{definition}
The set of ``successful'' QuBaR attacks is defined as
\begin{equation}
\mathcal{S}_p \eqdef
\left\{\mathfrak{A}=\{G_k\}_{k=1}^{\mathcal{C}\left( \fA \right)} :
\Pr\left(\vee_{k=1}^{\mathcal{C}\left( \fA \right)}
\left[\cT\left(G_k \right)=1\right]\right)>1/2\right\}.
\label{eq:attack-class}
\end{equation}
\label{def:attack-class}
\end{definition}
We treat any $\mathfrak{A} \in \mathcal{S}_p$ as
a {\em successful QuBaR attack} and derive an achievable (to the first
order in the exponent) lower bound on the complexity of these
attacks.

First, for the sake of completeness, we state the achievability
result via providing an extended version of the proof regarding the
complexity of the proposed ``most probable choice attack'' of
\cite{ABSG, gou-sib:06} using our notation.
\begin{theorem}
{\em \cite{ABSG, gou-sib:06} (Achievability - General Case)} Under the assumptions
A1, A2, A3, A4, mentioned in Section~\ref{ssec:assumptions}, there
exists a QuBaR attack algorithm $\fA_{ach,opt} \in \cS_p$ against
ABSG with $\cC \left( \fA_{ach,opt} \right) = 2^{ L / 2}$.
Further, $\cC_{ave} \left( \fA_{ach,opt} \right) =  \frac{1}{2} \left( 2^{L/2}+1\right)$
where $\cC_{ave} \left( \fA_{ach,opt} \right)$ is the {\em
expected} complexity of $\fA_{ach,opt}$ over the probability
distribution induced by $\bq$.
\label{thrm:achievability-general}
\end{theorem}
\begin{proof}
See Appendix~\ref{app-6}.
\end{proof}

\begin{remark}
For the proposed attack $\fA_{ach,opt}$,
assuming that the generation
of all of the guesses $G \left( \cdot \right)$  and the
corresponding check algorithm $\cT \left( G \left( \cdot \right) \right)$
are $poly \left( L \right)$,
both the time and data complexity of the attack
can be shown to be equal to $2^{L/2}$ to the first order in the exponent.
\end{remark}

Next, we state the converse theorem, which can be viewed as a fundamental
result due to its {\em negative} nature as far as cryptanalysis concerned, to the
best of our knowledge.

\begin{theorem}
{\em (Converse - General Case)} Under the assumptions A1,
A2, A3, A4,  and for any $\fA \in \cS_p$, we necessarily have $\cC
\left( \fA \right) > \underline{\cC}_{min} \eqdef 2^{L/2 - 1}$.
\label{thrm:converse-general}
\end{theorem}

\begin{proof}
See Appendix~\ref{app-7}.
\end{proof}

Theorems~\ref{thrm:achievability-general}
and~\ref{thrm:converse-general} imply the following result:
\begin{corollary}
Under the assumptions A1, A2, A3, A4, the tight (to the first order
in the exponent) lower bound on algorithmic complexity of any QuBaR
attack against ABSG is $2^{L/2}$:
\[
\cC \left( \fA_{ach,opt} \right) \firstorder \cC_{ave} \left( \fA_{ach,opt} \right)
\firstorder \underline{\cC}_{min}.
\]
\label{cor:general-lower-bound}
\end{corollary}

\begin{remark}
A practically useful consequence of Corollary~\ref{cor:general-lower-bound}
is as follows: In order to develop a successful ``query-based-recovery'' (QuBaR)
attack (in the sense of being an element of $\cS_p$) of complexity less than
$\cO \left( 2^{L/2} \right)$ (say $poly \left( L \right)$), it is necessary
to consider a construction where at least one of the assumptions A1, A2, A3, A4
is relaxed. Recalling these assumptions, it is advisable to concentrate on a
setup where the assumptions A1 and/or A3 do not apply; in practice, this may
lead to using a deterministic approach \cite{Alt08}, where explicit knowledge of the
generating LFSR's feedback polynomial is utilized and the input sequence to ABSG,
$\bx$, is an $M$-sequence\footnote{For further details on $M$-sequences, we refer
the interested reader to \cite{golomb}.}.
\label{rem:attack-advice}
\end{remark}

\section{Conclusion}
\label{sec:conclusion}

In this paper, we introduce a novel approach to cryptanalysis. We
aim to explore {\em fundamental performance limits} within a
specified class of attacks of interest, targeted towards breaking a
particular cryptosystem. As a first step, we illustrate our approach
via considering the class of ``Query-Based Key-Recovery'' (QuBaR)
attacks against ABSG, which is an LFSR-based stream cipher
constructed via irregular decimation techniques. In order to achieve
this task, we rely on the following assumptions (which are quite
common in conventional cryptanalysis): The input sequence to ABSG is
assumed to be an independent identically distributed Bernoulli
process with probability $1/2$; the attacker has access to the
output sequence of ABSG; an explicit knowledge of the generating
LFSR's feedback polynomial is not used; and the degree of the
feedback polynomial (denoted by $L$) of the generating LFSR is
sufficiently large. Using these assumptions, we show that  breaking
ABSG is equivalent to determine the exact realizations of a sequence
of random variables, which are proven to be independent identically
distributed with geometric distribution of parameter $1/2$. Next, we
investigate two setups of interest. In the first setup, we
concentrate on the ``Exhaustive-Search Type QuBaR'' attacks (which
form a subset of general-case QuBaR attacks, such that the  starting
index of all guesses in any element of this set is constrained to be
equal to unity). Here, using notions from information theory (in
particular asymptotic equipartition property \cite{sha:48}), we
prove that the tight lower bound (to the first order in the
exponent) on the algorithmic complexity of any successful
Exhaustive-Search Type QuBaR attack is $2^{2 L/3}$. In the second
setup, we concentrate on the general case QuBaR attacks and follow
an analogous development to that of the former setup. In particular,
we prove that the tight lower bound (to the first order in the
exponent) on the algorithmic complexity of any successful  QuBaR
attack is $2^{L/2}$. Our results can be viewed as a ``negative
advice'' to the cryptanalyst (contrary to the conventional trend in
cryptanalysis, where the general goal is to deduce a ``negative
design advice'' to the cryptosystem designer) in terms of QuBaR
attacks against ABSG under the aforementioned assumptions.

\appendices

\section{Proof of Lemma~\ref{lem:Q-pmf}}
\label{app-1}
\setcounter{equation}{0}
\renewcommand{\theequation}{I-\arabic{equation}}

First, note that each output bit $Z_i = z_i$ (for $1 \leq i \leq N$) is produced by a {\em block}
of input bits from the input sequence $\bX_1^M$. In order to identify the $i$-th input block that
generates $Z_i$ (for $1 \leq i \leq N$), we define
\begin{eqnarray}
A_i & \eqdef & 1 + \sum_{j=1}^{i-1} \left[ Q_j + 2 \right]  = H_{i-1} - H_0 + 1  = H_{i-1} + 1 , \nonumber \\
B_i & \eqdef & \sum_{j=1}^i \left[ Q_j + 2 \right]  = H_{i} - H_0 =
H_i  , \nonumber
\end{eqnarray}
where we used $H_0 = 0$ as the initial condition. Hence, we note
that the input block $\bX_{A_i}^{B_i}$ produces the $i$-th output
bit $Z_i = z_i$ which is given per assumption A2. Further, from the
definition of the algorithm $\cB$ (see Definition~\ref{def:ABSG}),
we have
\begin{equation}
\Pr \left( X_{A_i+1} = z_i \, | \, Z_i = z_i \right) = 1 .
\label{eq:lem1fund}
\end{equation}

Next, note that the statement of the lemma is {\em equivalent to}
\begin{equation}
\Pr \left( \bQ_1^N = \bq_1^N \, | \, \bZ_1^N = \bz_1^N \right) =
\prod_{i=1}^N \left[ \Pr \left( Q_i = q_i \, | \, \bZ_1^N = \bz_1^N \right) \right] =
\prod_{i=1}^N \left( \frac{1}{2} \right)^{q_i+1} .
\label{eq:lem1proof1}
\end{equation}
Thus, it is necessary and sufficient to show (\ref{eq:lem1proof1}) to prove Lemma~\ref{lem:Q-pmf}.
In order to show (\ref{eq:lem1proof1}), we use proof by induction.
\begin{itemize}
\item \underline{Step 1:}
We would like to show
\begin{equation}
\Pr \left( Q_1 = q_1  \, | \, \bZ_1^N = \bz_1^N \right) = \left( \frac{1}{2} \right)^{q_1+1} .
\label{eq:lem1proof2}
\end{equation}
Since the value of $Q_1$ depends only on the first output bit, we have
\[
\Pr \left( Q_1 = q_1  \, | \, \bZ_1^N = \bz_1^N \right) = \Pr \left( Q_1 = q_1  \, | \, Z_1 = z_1 \right) .
\]
Next,
\begin{eqnarray}
\Pr \left( Q_1 = 0 \, | \, Z_1 = z_1 \right) & = & \Pr \left( X_1 = z_1 , X_2 = z_1 \, | \, Z_1 = z_1 \right) ,
\label{eq:lem1proof3} \\
& = & \Pr \left( X_1 = z_1 \, | \, Z_1 = z_1 \right)  ,
\label{eq:lem1proof4} \\
& = & \frac{1}{2} ,
\label{eq:lem1proof5}
\end{eqnarray}
where (\ref{eq:lem1proof3}) follows from the definition of the
mapping ${\mathcal M(\cdot,\cdot)}$ (Table~\ref{tab:mapping}),
(\ref{eq:lem1proof4}) follows from (\ref{eq:lem1fund}),
(\ref{eq:lem1proof5}) follows from assumption A1. Also, for $q_1 >
0$,
\begin{eqnarray}
\Pr \left( Q_1 = q_1 \, | \, Z_1 = z_1 \right) & = &
\Pr \left( X_1 = \bar{z}_1 , X_2 = z_1 , \ldots , X_{q_1+1} = z_1 , X_{q_1+2} = \bar{z}_1 \, | \, Z_1 = z_1 \right) ,
\label{eq:lem1proof6}  \\
& = & \Pr \left( X_1 = \bar{z}_1 , X_3 = z_1 , \ldots , X_{q_1+1} = z_1 , X_{q_1+2} = \bar{z}_1 \, | \, Z_1 = z_1 \right) ,
\label{eq:lem1proof7}  \\
& = & \left( \frac{1}{2} \right)^{q_1 + 1}
\label{eq:lem1proof8}
\end{eqnarray}
where (\ref{eq:lem1proof6}) follows from the definition of the mapping ${\mathcal M} \left( \cdot , \cdot \right)$ (Table~\ref{tab:mapping}),
(\ref{eq:lem1proof7}) follows from (\ref{eq:lem1fund}), (\ref{eq:lem1proof8}) follows from assumption A1. Combining
(\ref{eq:lem1proof5}) and (\ref{eq:lem1proof8}), we get (\ref{eq:lem1proof2}).
\item \underline{Step 2:} We assume that
\begin{equation}
\Pr \left( \bQ_1^{n-1} = \bq_1^{n-1} \, | \, \bZ_1^N = \bz_1^N \right) =
\prod_{i=1}^{n-1} \left[ \Pr \left( Q_i = q_i \, | \, \bZ_1^N = \bz_1^N \right) \right] =
\prod_{i=1}^{n-1} \left( \frac{1}{2} \right)^{q_i+1} .
\label{eq:lem1inductionstep}
\end{equation}
\item \underline{Step 3:} Given (\ref{eq:lem1inductionstep}) we want to show that
\begin{equation}
\Pr \left( \bQ_1^{n} = \bq_1^{n} \, | \, \bZ_1^N = \bz_1^N \right) =
\prod_{i=1}^{n} \left[ \Pr \left( Q_i = q_i \, | \, \bZ_1^N = \bz_1^N \right) \right] =
\prod_{i=1}^{n} \left( \frac{1}{2} \right)^{q_i+1} .
\label{eq:lem1mainstep}
\end{equation}
Note that, given (\ref{eq:lem1inductionstep}), (\ref{eq:lem1mainstep}) is equivalent to
\begin{equation}
\Pr \left( Q_n = q_n \, | \, \bQ_1^{n-1} = \bq_1^{n-1} , \bZ_1^N = \bz_1^N \right) =
\Pr \left( Q_n = q_n \, | \, \bZ_1^N = \bz_1^N \right) = \left( \frac{1}{2} \right)^{q_n+1} ,
\label{eq:lem1proof9}
\end{equation}
using Bayes rule. Now,
\begin{eqnarray}
\Pr \left( Q_n = 0 \, | \, \bQ_1^{n-1} = \bq_1^{n-1} , \bZ_1^N = \bz_1^N \right)  & = &
\Pr \left( X_{A_n} = z_n , X_{A_n+1} = X_{B_n} = z_n \, | \, \bQ_1^{n-1} = \bq_1^{n-1} , \bZ_1^N = \bz_1^N \right) ,
\label{eq:lem1proof10} \\
& = & \Pr \left( X_{A_n} = z_n  \, | \, \bQ_1^{n-1} = \bq_1^{n-1} , \bZ_1^N = \bz_1^N \right) ,
\label{eq:lem1proof11} \\
& = & \Pr \left( X_{A_n} = z_n  \, | \, \bZ_1^N = \bz_1^N \right)  = \Pr \left( Q_n = 0 \, | \, \bZ_1^N = \bz_1^N \right) ,
\label{eq:lem1proof12} \\
& = & \Pr \left( X_{A_n} = z_n  \, | \, Z_n = z_n \right) = \frac{1}{2}
\label{eq:lem1proof13}
\end{eqnarray}
where (\ref{eq:lem1proof10}) follows from the definition of the
mapping ${\mathcal M(\cdot,\cdot)}$ (Table~\ref{tab:mapping}),
(\ref{eq:lem1proof11}) follows from (\ref{eq:lem1fund}),
(\ref{eq:lem1proof12}) and (\ref{eq:lem1proof13}) follow from
assumption A1 \footnote{Since $\bX$ is an i.i.d. Bernoulli $1/2$
process,  the value of $\Pr \left( X_{A_n} = z_n  \, | \, Z_n = z_n
\right) $ is independent of the particular value of $A_n$ and that
is why it is equal to $1/2$.}. On the other hand, for $q_n > 0$, we
have
\begin{eqnarray}
& & \Pr \left( Q_n = q_n \, | \, \bQ_1^{n-1} = \bq_1^{n-1} , \bZ_1^N = \bz_1^N \right)  \nonumber \\
& = & \Pr \left( X_{A_n} = \bar{z}_n , X_{A_n+1} = z_n , \ldots X_{B_n-1} = z_n, X_{B_n} = \bar{z}_n
\, | \, \bQ_1^{n-1} = \bq_1^{n-1} , \bZ_1^N = \bz_1^N \right) ,
\label{eq:lem1proof14} \\
& = & \Pr \left( X_{A_n} = \bar{z}_n , X_{A_n+2} = z_n , \ldots X_{B_n-1} = z_n, X_{B_n} = \bar{z}_n
\, | \, \bQ_1^{n-1} = \bq_1^{n-1} , \bZ_1^N = \bz_1^N \right) ,
\label{eq:lem1proof15} \\
& = & \Pr \left( X_{A_n} = \bar{z}_n , X_{A_n+2} = z_n , \ldots X_{B_n-1} = z_n, X_{B_n} = \bar{z}_n
\, | \, \bZ_1^N = \bz_1^N \right) ,
\label{eq:lem1proof16} \\
& = & \Pr \left( X_{A_n} = \bar{z}_n , X_{A_n+2} = z_n , \ldots X_{B_n-1} = z_n, X_{B_n} = \bar{z}_n
\, | \, Z_n = z_n \right)  = \left( \frac{1}{2} \right)^{q_n+1},
\label{eq:lem1proof17}
\end{eqnarray}
where (\ref{eq:lem1proof14}) follows from the definition of the
mapping ${\mathcal M(\cdot,\cdot)}$ (Table~\ref{tab:mapping}),
(\ref{eq:lem1proof15}) follows from (\ref{eq:lem1fund}),
(\ref{eq:lem1proof16}) and (\ref{eq:lem1proof17}) follow from
assumption A1 (see the discussion in the footnote). Combining
(\ref{eq:lem1proof12}), (\ref{eq:lem1proof13}),
(\ref{eq:lem1proof16}), (\ref{eq:lem1proof17}), we get
(\ref{eq:lem1proof9}), and equivalently (\ref{eq:lem1mainstep}),
which completes the proof.
\end{itemize}
\qed

\section{Proof of Theorem~\ref{theorem:attack-equivalence}}
\label{app-2}

\setcounter{equation}{0}
\renewcommand{\theequation}{II-\arabic{equation}}

\setcounter{lemma}{0}
\renewcommand{\thelemma}{II-\arabic{lemma}}

The equivalence of the first and second problems is shown in
\cite{gou-sib:06}. In order to prove the theorem, we proceed with
proving the equivalence of the second and third problems.

First, we show that the third problem reduces to the second problem in $poly \left( L \right)$ time:
Since we know $\bz_1^N$ and $\bQ_{i}^{\theta+i-1}$ per assumption, we construct $L$ consecutive
bits of $\bX_1^M$ via using Definition~\ref{def:ABSG} in the following way. We are given
$\bQ_{i}^{i + \theta -1} = \bq_{i}^{i + \theta -1}$ such that (\ref{eq:pbm-equivalence-constraint}) holds.
Then, we apply the following algorithm:
\begin{enumerate}
\item For each $j = i , i+1 , \ldots , \theta + i - 1$ do:
\begin{enumerate}
\item If $q_j = 0$, generate $B_j = \left\{ z_j , z_j \right\}$.
\item If $q_j > 0$, generate $B_j = \left\{ \bar{z}_j, z_j^{q_j} , \bar{z}_j \right\}$
\end{enumerate}
\item Concatenate $\left\{ B_j \right\}_{j=i}^{\theta + i - 1}$ thereby forming the desired $\bX = \bx$ sequence.
\end{enumerate}
Note that, the condition (\ref{eq:pbm-equivalence-constraint}) ensures that the resulting $\bX  = \bx $ sequence
$\left\{ B_i , B_{i+1} , \ldots , B_{\theta+i-1} \right\}$ is of length at least $L$. Furthermore, from the
definition of the ABSG algorithm, the resulting $\bX = \bx$ sequence is unique and necessarily the correct one.
Obviously, this algorithm runs in $poly \left( L \right)$ time, which completes the proof for this case.

Next, we proceed with showing that the second problem can be reduced to the
third problem via an algorithm in probabilistic polynomial time.
First, note the following Lemma.
\begin{lemma}
Under the assumptions A1, A2, A3, and A4, for any $n \in \bbZ^+$, we have
\begin{equation}
\Pr \left[ \wedge_{l=0}^{poly \left( L \right)} Y_{n+l} \neq \varnothing \right] \leq \epsilon ,
\label{eq:ABSGPr-expsmall}
\end{equation}
for any $\epsilon > 0$ for $L$ sufficiently large.
\label{lem:cirkin}
\end{lemma}
\begin{proof}
First, under the given assumptions, we note the following fundamental results from \cite{Alt08}:
\begin{itemize}
\item For any $n \in \bbZ^+$,
\begin{equation}
\Pr \left( Y_n = \varnothing \right) = \frac{1}{3} + \frac{2}{3} \left( - \frac{1}{2} \right)^n .
\label{eq:temp1}
\end{equation}
\item $\left\{ Y_n \right\}$ form a Markovian process with memory-$1$:
\begin{equation}
\Pr \left( Y_n | \bY_1^{n-1}  = \by_1^{n-1} \right) = \Pr \left( Y_n | Y_{n-1} = y_{n-1} \right) ,
\label{eq:temp2}
\end{equation}
for any $n \in \bbZ^+$.
\item For any $n \in \bbZ^+$,
\begin{equation}
\Pr \left( Y_n \neq \varnothing | Y_{n-1} \neq \varnothing \right) = \frac{1}{2}.
\label{eq:temp3}
\end{equation}
\end{itemize}
Hence, for any $\epsilon >0$ we have
\begin{eqnarray}
\Pr \left[ \wedge_{l=0}^{poly \left( L \right)} Y_{n+l} \neq \varnothing \right] & = &
\Pr \left( Y_n \neq \varnothing \right) \cdot \prod_{l=1}^{poly \left( L \right)}
\Pr \left( Y_{n+l} \neq \varnothing | \wedge_{k=0}^{l-1} Y_{n+k} \neq \varnothing \right) , \label{eq:temp4} \\
& = & \Pr \left( Y_n \neq \varnothing \right) \cdot \prod_{l=1}^{poly \left( L \right)}
\Pr \left( Y_{n+l} \neq \varnothing | Y_{n+l-1} \neq \varnothing \right) , \label{eq:temp5} \\
& = & \left[ \frac{2}{3} - \frac{2}{3} \left( - \frac{1}{2} \right)^n \right] \cdot \left( \frac{1}{2} \right)^{poly \left( L \right) - 1} ,
\label{eq:temp6} \\
& \leq & \epsilon \label{eq:temp7}
\end{eqnarray}
where (\ref{eq:temp4}) follows from Bayes rule, (\ref{eq:temp5}) follows from (\ref{eq:temp2}),
(\ref{eq:temp6}) follows from (\ref{eq:temp1}) and (\ref{eq:temp3}), (\ref{eq:temp7}) follows from the fact that
the first term in (\ref{eq:temp6}) is constant in $L$ and the second term is exponentially decaying in $L$ whence
$\epsilon$ can be made arbitrarily small for sufficiently large $L$.
\end{proof}
Now, since $Y_n \in \left\{ 0 , 1, \varnothing \right\}$ (i.e., there are constant possibilities for $Y_n$), w.l.o.g. 
we assume that $Y_n$ is known. Since we are also given $\bX_{n+1}^{n+L}$ for some $n \in \bbZ^+$,
this also means we know $\bY_n^{n+L}$ (via successively applying $\mathcal{M} \left( Y_{n+l-1} , X_{n+l} \right)$ for
$l=1,2,\ldots , L$). Next, consider the following situations:
\begin{enumerate}
\item \underline{$Y_n = Y_{n+L} = \varnothing$}: \\
In this case, w.l.o.g. we choose $h_{i-1} = n$ for some $i$. Next,
let $K$ denote the number of $\varnothing$'s within the sequence
$\bY_n^{n+L}$ (which is necessarily $\geq 2$ per assumption) and
assign $\theta = K-1$. Next, let $h_{i+j-2}$ denote the index of the
$j$-th $\varnothing$ within the sequence $\bY_n^{n+L}$, where $1
\leq j \leq K = \theta+1$ (implying $h_{i+K-2} = h_{i+\theta-1} =
n+L$). Accordingly, assign $q_j = h_j - h_{j-1} - 2$ for all $j \in
\left\{ i , i+1, \ldots , i+\theta - 1 \right\}$. Note that, all
these $\left\{ h_j \right\}$ (equivalently $\left\{ q_j \right\}$)
are known since $\bY_n^{n+L}$ is known. Consequently, this means we
have identified $\bQ_i^{i+\theta-1} = \bq_i^{i+\theta-1}$ such that
\[
\sum_{j=i}^{\theta+i-1} \left( q_j + 2 \right) = \sum_{j=i}^{\theta+i-1} \left( h_j - h_{j-1} \right) = h_{\theta+i-1} - h_i  = L ,
\]
satisfying the constraint (\ref{eq:pbm-equivalence-constraint}). Further, note that the operations performed within this
procedure constitute an algorithm, which is in deterministic polynomial time (implying it is also in probabilistic polynomial time).
\item \underline{$Y_n = \varnothing$ and $Y_{n+L} \neq \varnothing$}: \\
In this case, since $Y_{n+L} \neq \varnothing$, we aim to identify some $Y_{n+L+L'} = \varnothing$ for $L'>0$
with high probability in polynomial time.
To achieve this task, we consider the sequence $\left\{ Y_{n+L+k} \right\}$ for $k>0$.
Now, note that as we increment $k$, after $poly \left( L \right)$ steps we necessarily need to come across
a $\varnothing$ with high probability (the probability of {\em not} coming across a $\varnothing$ is exponentially
small in $L$ per Lemma~\ref{lem:cirkin}). Thus, we have $Y_n = Y_{n+L''} = \varnothing$ where $L'' = L+L' > L$.
Next, applying algorithmic steps analogous to the ones in Situation 1 (i.e., beginning from  $Y_n = Y_{n+L''} = \varnothing$),
we identify $\bQ_i^{i+\theta-1} = \bq_i^{i+\theta-1}$  such that
\[
\sum_{j=i}^{\theta+i-1} \left( q_j + 2 \right) = \sum_{j=i}^{\theta+i-1} \left( h_j - h_{j-1} \right) = h_{\theta+i-1} - h_i  = L''>L ,
\]
satisfying the constraint (\ref{eq:pbm-equivalence-constraint}).
Further, note that the operations performed within this
procedure constitute an algorithm, which is in probabilistic polynomial time.
\item \underline{$Y_n \neq \varnothing$ and $Y_{n+L} = \varnothing$}: \\
Our overall goal is to identify (via using an algorithm, which is in probabilistic polynomial time) $Y_{n+L} = Y_{n+L'''} =
\varnothing$ such that $L''' - L > L$. In that case, we would be able to apply algorithmic steps analogous to the ones
in Situation 1 (i.e., beginning from $Y_{n+L} = Y_{n+L'''} = \varnothing$) and identify $\bQ_i^{i+\theta-1} =
\bq_i^{i+\theta-1}$  such that
\[
\sum_{j=i}^{\theta+i-1} \left( q_j + 2 \right) = \sum_{j=i}^{\theta+i-1} \left( h_j - h_{j-1} \right) = h_{\theta+i-1} - h_i  = L'''-L>L ,
\]
satisfying the constraint (\ref{eq:pbm-equivalence-constraint}).
Next, we show that, beginning from $Y_{n+L}$, we are able to find
some $Y_{n+L'''} = \varnothing$ such that $L''' > 2L$ via a
probabilistic polynomial time algorithm. To see this, first consider
the sequence $\left\{ Y_{n+L+k} \right\}$ for $k>0$ (as we did in
Situation 2). Following Lemma~\ref{lem:cirkin} and using similar
arguments to Situation 2, we see that as we increment $k$ by $poly
\left( L \right)$, we necessarily come across a $\varnothing$ with
high probability.  Next, we apply this step $L/2$ times; at each
step, we increment $k$ by $poly \left( L \right)$ and at each step,
we see a $\varnothing$ with probability $1 - \epsilon$ where
$\epsilon$ is exponentially small in $L$ per Lemma~\ref{lem:cirkin}.
Thus, as a result of incrementing $k$ by a total of $\frac{L}{2}
\cdot poly \left( L \right)$ (which is again $poly \left( L
\right)$), we observe $L/2$ $\varnothing$'s with sufficiently high
probability, which makes this procedure an algorithm in
probabilistic polynomial time. On the other hand, observing $L/2$
$\varnothing$'s guarantee us to identify some $L'''$ such that $L'''
> 2 L$ since the gap between two $\varnothing$'s is at least $2$ due
to the definition of the ABSG algorithm. As a result, we see that we
can identify $Y_{n+L'''} = \varnothing$ such that $L''' - L > L$ via
an algorithm which is in probabilistic polynomial time, which was
our initial goal.
\item \underline{$Y_n \neq \varnothing$ and $Y_{n+L}  \neq  \varnothing$}: \\
This is straightforward via applying an approach analogous to the Situation 3 above. Again, we begin from $Y_{n+L}$,
consider the sequence $Y_{n+L+k}$ for $k>0$, increment $k$ in blocks of length $poly \left( L \right)$; the only difference
is that this time we use $\frac{L}{2} + 1$ blocks (each of which is $poly \left( L \right)$) instead of $\frac{L}{2}$. As a result,
we are guaranteed to identify $Y_{n+L''''} = Y_{n+L'''''} = \varnothing$ such that $L''''' - L'''' > L$ via an algorithm which is
in probabilistic polynomial time; the rest is obvious.
\end{enumerate}
Thus, the proof of the (probabilistic polynomial time) reduction of the second problem to the third one is completed.
Hence the proof of Theorem~\ref{theorem:attack-equivalence}.
\qed

\section{Proof of Theorem~\ref{thrm:achievability-exhaustive}}
\label{app-3}
\setcounter{equation}{0}
\renewcommand{\theequation}{III-\arabic{equation}}

For the sake of clarity, throughout this section we use the notation
$G_k \left( i_k, \theta_k, \left( \bq_{i_k}^{\theta_k+i_k-1} \right)_k \right)$
(instead of $G_k \left( i_k, \theta_k,  \bq_{i_k}^{\theta_k+i_k-1}  \right)$)
to denote a particular guess $G_k$.

Choosing $n \eqdef L / 3$, first we define the typical set
$A_{\epsilon}^{\left( n \right)}$ with respect to $p \left( q \right)$ (given by (\ref{eq:Q-pmf})):
\begin{equation}
A_\epsilon^{\left( n\right)} \eqdef \left\{ \bq_1^n \, : \, \left| - \frac{1}{n}
\log p \left( \bq_1^n \right) - H \left( Q \right) \right| \leq \epsilon \right\} ,
\label{eq:typicalsetdefinition}
\end{equation}
where (using logarithm with base-$2$)
\[
H \left( Q \right) = - \sum_{q=0}^\infty p \left( q \right) \log p \left( q \right)
= 2.
\]
At this point, we also recall two fundamental results regarding typical sets \cite{cov:91}:
\begin{eqnarray}
\left( 1 - \epsilon \right) 2^{n \left( H \left( Q \right) - \epsilon \right)} \, \leq \,
\left| A_\epsilon^{\left(n\right)} \right| \, \leq \,
2^{n \left( H \left( Q \right) + \epsilon \right)}
& & \label{eq:typicalset-cardinality}\\
\Pr \left( \bq_1^n \in  A_\epsilon^{\left(n\right)}  \right) > 1 - \epsilon
& & \label{eq:typicalset-probability}
\end{eqnarray}
for any $\epsilon >0$, for sufficiently large $n$.

Next, we propose the following construction for the attack $\fA_{ach,opt}^E$: \\
{\bf 1.} Index all $\bq_1^n \in A_\epsilon^{\left( n \right)}$ and accordingly
let $\left( \bq_1^n \right)_k$ denote the $k$-th element where
$k \in \left\{ 1 , 2 , \ldots , \left| A_\epsilon^{\left( n \right)} \right| \right\}$.
Let $q_{i,k}$ denote the $i$-th element of $\left( \bq_1^n \right)_k$ for
$i \in \left\{ 1 , 2, \ldots , n \right\}$. \\
{\bf 2.} At each $k$-th step of the QuBaR attack, choose
$G_k = \left( i_k = 1 , \theta_k = n = \frac{L}{3} , \left( \bq_1^n \right)_k \right)$;
$k \in \left\{ 1 , 2 , \ldots , \left| A_\epsilon^{\left( n \right)} \right| \right\}$.

Note that, this attack qualifies as a ``QuBaR attack against ABSG'' only if
all of the aforementioned guesses satisfy the constraint
(\ref{eq:pbm-equivalence-constraint}). To see that this is satisfied for arbitrarily
small $\epsilon$, we observe (noting that $\beta_k = \sum_{i=1}^n q_{i,k}$)
\begin{equation}
\left| -\frac{1}{n} \log p \left( \left( \bq_1^n \right)_k \right) - H \left( Q \right) \right|
= \left| \left( 1 + \frac{\beta_k}{\theta_k} \right)- 2 \right| \leq \epsilon
\label{eq:temp1-proof1}
\end{equation}
where the equality follows from (\ref{eq:Q-pmf}), the definition of $\beta_k$ and using $\theta_k = n$,
the inequality follows from (\ref{eq:typicalsetdefinition}). Furthermore,
using $\theta_k = n  = L / 3$, after straightforward algebra (\ref{eq:temp1-proof1}) can be shown to be
equivalent to
\begin{equation}
L \left( 1 - \frac{\epsilon}{3} \right) \leq 2 \theta_k + \beta_k \leq L \left( 1 + \frac{\epsilon}{3} \right) .
\nonumber
\end{equation}
Since we can choose $\epsilon$ arbitrarily small, the aforementioned attack qualifies
as a QuBaR attack against ABSG as $\epsilon \rightarrow 0$.

Next, (\ref{eq:typicalset-probability}) implies that for large $n$
(equivalently for large $L$) $\mbox{Pr}_{succ} \left(
\mathfrak{A}_{ach,opt}^E \right) = \Pr \left( \bq_1^n \in
A_\epsilon^{\left(n\right)} \right)$ can be made arbitrarily close
to $1$ since we can choose $\epsilon$ arbitrarily small. Thus,
$\Pr_{succ} \left( \fA_{ach,opt}^E \right) \rightarrow 1$ as $L
\rightarrow \infty$ and $\epsilon \rightarrow 0$. Furthermore, for
large $L$ the algorithmic complexity is at most $\left|
A_\epsilon^{\left( n \right)} \right|$ which can be made arbitrarily
close to $2^{2L/3}$ per (\ref{eq:typicalset-cardinality}) since
$n=L/3$, $H \left( Q \right) = 2$ and we can choose $\epsilon$
arbitrarily small. Thus, the algorithmic complexity is at most $2^{2
L/3}$ as $L \rightarrow \infty$, $\epsilon \rightarrow 0$. Recalling
that for sufficiently small $\epsilon$, all elements of
$A_{\epsilon}^{(n)}$ are equiprobable (since $\beta_k \in \bbN,
\theta_k \in \bbZ^+$) with probability $\left. 2^{-(\theta_k + \beta_k)}
\right|_{\theta_k = \beta_k = L/3}$, we immediately see that the
expected algorithmic complexity is $\frac{1}{2} \left( 2^{2 L / 3} +
1 \right)$. Note that in the proposed attack, $i_k = 1$ and
$\theta_k = n = L /3$ for all $k$ which implies that the
corresponding data complexity is $L/3$. \qed


\section{Proof of Theorem~\ref{thrm:converse-exhaustive}}
\label{app-4}
\setcounter{definition}{0}
\renewcommand{\thedefinition}{IV-\arabic{definition}}
\setcounter{equation}{0}
\renewcommand{\theequation}{IV-\arabic{equation}}
\setcounter{lemma}{0}
\renewcommand{\thelemma}{IV-\arabic{lemma}}

First of all, since $L$ is sufficiently large (per assumption A4), we assume w.l.o.g.
$L$ is divisible by $6$. Our fundamental goal is to characterize the
algorithmic complexity of the optimal attacks subject to a lower
bound on the success probability of the attack. Thus,
we aim to analytically identify
\begin{equation}
\fA_{opt}^E \eqdef \arg \min_{\fA^E \in \cS_p^E} \cC \left( \fA^E
\right) , \label{eq:app4-1}
\end{equation}
where
\begin{equation}
\cS_p^E \eqdef \left\{ \fA^E \; : \; \fA^E \in \cS^E \, \mbox{and}
\, \Pr\left( \vee_{k=1}^{\cC \left( \fA^E \right)} \left[ \cT \left(
G_k \right) =1 \right] \right) > \frac{1}{2}  \right\} \subseteq
\cS^E, \nonumber
\end{equation}
i.e., $\cS_p^E$ is a ``probabilistically-constrained'' subset of $\cS^E$ for which the success probability
is strictly bounded away from $1/2$.
In our terminology, we denote the quantity of $\Pr\left( \vee_{k=1}^{\cC \left( \fA \right)} \left[ \cT \left( G_k \right) =1 \right] \right)$
as the {\em success probability of algorithm $\fA$}.
Our problem is to characterize
\[
\cC_{min}^E \eqdef \cC \left( \fA_{opt}^E \right) ,
\]
in particular, we aim to achieve this goal via quantifying a lower bound on it.

Our proof approach can be summarized as follows: Since it is not a
straightforward task to solve the optimization problem
(\ref{eq:app4-1}), we proceed with a simpler problem. We define a
set $\tcS_p^E$, such that $\cS_p^E \subseteq \tcS_p^E \subseteq
\cS^E$, and accordingly proceed with minimizing $\cC \left( \fA^E
\right)$ over all $\fA^E \in \tcS_p^E$. 
The set $\tcS_p^E$ is defined in such a way that 
minimizing $\cC \left( \fA^E \right)$ over this set (i.e., over all $\fA^E \in \tcS_p^E$)
is tractable. 
At the last
step, we conclude the proof via deriving a lower bound on the
minimum algorithmic complexity over $\tcS_p^E$, which also forms a
lower bound on $\cC_{min}^E$ since $\cS_p^E \subseteq \tcS_p^E$.

We proceed with defining the set
\begin{equation}
\tcS_p^E \eqdef \left\{ \fA^E \; : \; \fA^E \in \cS^E \, \mbox{and}
\, \sum_{k=1}^{\cC \left( \fA^E \right)} \Pr \left( \cT \left( G_k
\right) =1 \right) > \frac{1}{2}  \right\} \nonumber
\end{equation}
In our terminology, we denote the quantity of $\sum_{k=1}^{\cC \left( \fA \right)} \Pr \left( \cT \left( G_k \right) =1 \right)$
as the {\em cumulative success probability of algorithm $\fA$}.
Note that, success probability is always upper-bounded by cumulative success probability
for any algorithm $\fA$; i.e., we have
\[
\Pr\left( \vee_{k=1}^{\cC \left( \fA \right)} \left[ \cT \left( G_k \right) =1 \right] \right)
\leq
\sum_{k=1}^{\cC \left( \fA \right)} \Pr \left( \cT \left( G_k \right) =1 \right)
\]
due to the union bound, which implies $\cS_p^E \subseteq \tcS_p^E \subseteq \cS^E$.
Next, we define the optimization problem (which is ``alternate'' to (\ref{eq:app4-1}))
\begin{equation}
\tfA_{opt}^E \eqdef \arg \min_{\fA^E \in \tcS_p^E} \cC \left( \fA^E
\right) , \label{eq:app4-alternate-pbm}
\end{equation}
and accordingly
\[
\tcC_{min}^E \eqdef \cC \left( \tfA_{opt}^E \right).
\]

In order to quantify the solution of (\ref{eq:app4-alternate-pbm}),
for the sake of convenience we define
\begin{equation}
\cG \left( \theta , \alpha \right) \eqdef \left\{ \bq_1^\theta \; :
\; \forall i, \, q_i \geq 0, \, \theta \in \bbZ^+, \, \alpha \in
\bbN, \, \sum_{i=1}^\theta q_i = \beta = L - 2 \theta + \alpha
\right\} \label{eq:app4-Gset-def}
\end{equation}
for any given $\theta \in \bbZ^+$ and $\alpha \in \bbN$. Observe
that  $\left\{ \cG \left( \theta , \alpha \right) \right\}$ are clearly disjoint
for different pairs of $\left\{ \left( \theta , \alpha \right) \right\}$.
Further, note that, by construction,
$\bq_1^\theta \in \cG \left( \theta , \alpha \right)$ for some $\theta \in
\bbZ^+$, $\alpha \in \bbN$ implies
(\ref{eq:pbm-equivalence-constraint}) since $2 \theta + \beta = L +
\alpha \geq L$; thus, any guess $G  = \left( 1 , \theta ,
\bq_1^\theta \right)$ where $\bq_1^\theta \in \cG \left( \theta ,
\alpha \right)$ for some $\theta \in \bbZ^+$, $\alpha \in \bbN$ is a
valid ABSG-guess. Furthermore, any valid guess $G$ necessarily
corresponds to a $\bq_1^\theta \in \cG \left( \theta , \alpha \right)$ for some unique pair
$\left( \theta , \alpha \right)$. Next, using \eqref{eq:pmf-I} observe that
\begin{equation}
p\left( \bq_1^{\theta}\right) = \left. 2^{-\left( \theta + \beta \right)} \right|_{\beta = L - 2  \theta + \alpha }
 = 2^{- \left( L - \theta + \alpha \right)},
\label{eq:app4-Gset-prob}
\end{equation}
for any $\bq_1^{\theta} \in \cG \left( \theta, \alpha \right)$;
i.e., given a pair $\left( \theta , \alpha \right)$, all elements of $\cG \left( \theta , \alpha \right)$
are equally likely with probability $2^{- \left( L - \theta + \alpha \right)}$.

Going back to \eqref{eq:app4-alternate-pbm}, since we are trying to
achieve a cumulative success probability strictly greater than $1/2$
using elements from {\em disjoint} sets $\left\{ \cG \left( \theta ,
\alpha \right) \right\}$, the optimal strategy is clearly to use the
{\em sorted} elements $\bq_1^{\theta} \in \cG \left( \theta , \alpha
\right)$ with respect to their success probabilities, specified in
\eqref{eq:app4-Gset-prob} \footnote{This problem is trivially
equivalent to the problem of obtaining a pre-specified amount of
cake with minimum number of slices, where the slice sizes are fixed,
but not necessarily uniform.}. Thus, algorithmically the optimal
solution consists of trying the guess with largest marginal success
probability first, and then the most probable guess in the remaining
ones, and so on.

Next, we aim to characterize the aforementioned sorting process and
analyze the minimum number of elements needed to achieve a
cumulative success probability strictly greater than $1/2$. Since
all elements of $\cG \left( \theta, \alpha \right)$ are equally
likely (cf. \eqref{eq:app4-Gset-prob}), the problem of sorting
individual sequences reduces to the problem of sorting the sets
$\left\{ \cG \left( \theta, \alpha \right) \right\}$ in
non-increasing order with respect to \eqref{eq:app4-Gset-prob}. The
total number of elements in these sorted sets of $\left\{ \cG \left(
\theta , \alpha \right) \right\}$ such that the total probability
exceeds $1/2$ amounts to the sought result $\tcC_{min}^E$. As a
result, we should solve the following sorting problem:

\textbf{Sorting Problem I:} Sort over $\left( \theta, \alpha
\right)$, with respect to the cost function $L - \theta + \alpha$,
in non-decreasing order, such that
\begin{equation}
\left( \theta, \alpha \right) \in \cS_{E,F} \eqdef \left\{ \left( \theta, \alpha
\right)\; : \; \theta \in \bbZ^+, \; \alpha \in \bbN, \; \beta = L-2 \theta
+ \alpha \geq 0\right\}.
\label{eq:app4-sorting-pbm1}
\end{equation}
Since this sorting needs to be done over $\left( \theta, \alpha
\right)$, our next task is to characterize the feasible set $\cS_{E,F}$ over which
the sorting will be carried out.

First of all, notice that from the definition of $\cG \left( \theta,
\alpha\right)$ (cf. \eqref{eq:app4-Gset-def}), we have
\begin{equation}
2 \theta - \alpha \leq L,
\label{eq:app4-temp1}
\end{equation}
since $\beta = L -2 \theta + \alpha \geq 0$.
Next, we define
\begin{equation}
B \eqdef L - \theta + \alpha,
\label{eq:app4-B}
\end{equation}
as our cost function in the aforementioned Sorting Problem I.
Note that, for any $\bq \in \cG \left( \theta, \alpha \right)$, $\Pr \left( \bQ = \bq \right) =
2^{-\left( \theta + \beta \right)} = 2^{-B}$; i.e., for any guess $G \left( i , \theta , \bq \right)$,
its success probability is equal to $2^{-B}$ where $B$ is computed via \eqref{eq:app4-B}
using the corresponding $\theta$ and $\alpha$. This means that, for any given guess $G \left( \cdot \right)$,
its marginal success probability, $ \Pr \left( \cT \left( G \right) = 1 \right)$ is directly determined
by the corresponding value of $B$.

Next, our goal is to find an alternate re-parameterized expression for (\ref{eq:app4-sorting-pbm1})
in terms of $B$ and $L$ since $B$ is our cost function in Sorting Problem I.
Now, using \eqref{eq:app4-B} in \eqref{eq:app4-temp1}  and noting that $\alpha \in \bbN$ yields
\begin{equation}
\alpha \in \left\{ 0, 1, \ldots, 2B - L\right\} .
\label{eq:app4-2}
\end{equation}
which also implies that $B \geq L/2$ since $\alpha \geq 0$.
As a side result, this accordingly implies the following upper
bound on the marginal success probability of any valid guess:
\begin{equation}
\Pr \left[ \cT \left( G \left( i , \theta , \bq \right) \right) = 1 \right]
= \left. 2^{-B} \right|_{B = \theta + \beta} \leq 2^{- L /2} \quad
\mbox{for any $G \left( i , \theta, \bq \right) \in \cG \left( \theta , \alpha \right)$
for some $\alpha \in \bbN$}.
\label{eq:upperboundonsuccess}
\end{equation}
The result \eqref{eq:upperboundonsuccess} will be useful in proving Theorem~\ref{thrm:converse-general}
of Section~\ref{sec:general-case}.

Next, per (\ref{eq:app4-B}),
each value of $\alpha$ uniquely determines $\theta$ in terms of $B$ via
\begin{equation}
\theta = L - B + \alpha.
\label{eq:app4-thetaBalpha}
\end{equation}
Using \eqref{eq:app4-thetaBalpha} in (\ref{eq:app4-2}), we have
\begin{equation}
\theta \in \left\{ L-B, L-B+1, \ldots, B \right\} ,
\label{eq:app4-3}
\end{equation}
which also implies that $B \leq L -1$ since $\theta \geq 1$. Combining these observations,
we find out the following equivalent expression to (\ref{eq:app4-sorting-pbm1}):
\begin{equation}
\left( \theta, \alpha \right) \in \cS_{E,F} =
\bigcup_{B = \frac{L}{2}}^{L-1}
\left\{ \left(  L - B  , 0 \right) , \left(  L - B + 1 , 1\right) ,
\left( L - B + 2 , 2\right) ,
\ldots ,
\left( B-1 , 2B-L-1\right) ,
\left( B , 2 B - L  \right) \right\} ,
\label{eq:app4-sorting-pbm2}
\end{equation}
where we effectively did a re-parameterization using $B$.
Note that, this re-parameterization allows us to see that,
given a fixed $B$, all $\left\{ \cG \left( \theta , \alpha \right) \right\}$
such that
\[
\left( \theta , \alpha \right) \in
\left\{ \left(  L - B  , 0 \right) , \left(  L - B + 1 , 1\right) ,
\left( L - B + 2 , 2\right) ,
\ldots ,
\left( B-1 , 2B-L-1\right) ,
\left( B , 2 B - L  \right) \right\} ,
\]
are equivalent to each other in terms of their success probabilities, $2^{-B}$.
Using this observation and \eqref{eq:app4-sorting-pbm2}, we conclude that Sorting Problem I is
equivalent to to the following one:

\textbf{Sorting Problem II:} Sort over $\left( B, \alpha \right)$
with respect to $B$ in non-decreasing order, such that
\begin{equation}
\left( B, \alpha\right) \in \left\{ (B,\alpha) \; : \;
\alpha \in \{0, 1, \ldots, 2B - L\}, \; B \in \{ L/2,
\ldots, L-1\}\right\}.
\label{eq:app4-sorting-pbm3}
\end{equation}

Note that, the corresponding values of $\theta$ in \eqref{eq:app4-sorting-pbm3} are
given by \eqref{eq:app4-thetaBalpha}.

Following is one of the solutions to Sorting Problem II:
\begin{equation}
\left\{ \left( B, \alpha \right)\right\} = \left\{ \left( L/2 , 0 \right), \,
\left( L/2 + 1 , 0 \right), \, \left( L / 2 + 1 , 1 \right) , \, \left( L / 2 + 1 , 2 \right) , \,
\left( L / 2 + 2 , 0 \right) \, \ldots, \left( L - 1 , L - 2 \right) \right\}.
\label{eq:app4-sorting-ex}
\end{equation}
Note that all solutions to Sorting Problem II are equivalent to each other
in terms of the resulting complexity. In particular, for a given $B$, we follow the strategy of
varying $\alpha$ in increasing order, beginning from $0$, ending in
$2 B - L$ as illustrated in \eqref{eq:app4-sorting-ex}.

Next, we concentrate on the range of $L/2 \leq B < 2 L / 3$ and analyze
the corresponding cumulative success probability (denoted by $P_1$) of the aforementioned strategy (cf.
\eqref{eq:app4-sorting-ex}), i.e.,
\begin{equation}
P_1 \eqdef \sum_{B=\frac{L}{2}}^{\frac{2L}{3}-1} \sum_{\alpha=0}^{2B-L}
\Pr\left( \left. \cG\left( \theta, \alpha \right) \right|_{\theta = L-B+\alpha} \right).
\label{eq:app4-succ-prob}
\end{equation}
Next, we derive an upper bound on $P_1$ which will be used in the subsequent computations.

\begin{lemma}
The cumulative success probability in the range of $L/2 \leq B < 2 L /3$ (i.e., $P_1$) is upper-bounded by
\begin{equation}
P_1 \leq \sum_{\theta = \frac{L}{3}+1}^{\frac{2L}{3}-1}
\sum_{\alpha=0}^{\theta - \frac{L}{3}-1} \Pr\left( \cG \left(
\theta, \alpha\right) \right).
\label{eq:app4-P1upperbound}
\end{equation}
\label{lem:app4-P1upperbound}
\end{lemma}
\begin{proof}
From \eqref{eq:app4-succ-prob}, we see that $P_1$ is defined in the $\left( B , \alpha \right)$ space (where $B = L - \theta + \alpha$), over the
set
\begin{equation}
\Lambda \eqdef \left\{ \left( B , \alpha \right) \; : \: \frac{L}{2} \leq B \leq \frac{2 L}{3} - 1 , \,
0 \leq \alpha \leq 2 B - L \right\} ,
\label{eq:app4-lemP1upperboundLambda}
\end{equation}
i.e., $P_1 = \sum_{ \left( B , \alpha \right) \in \Lambda} \left.
\Pr\left( \cG \left( \theta, \alpha\right) \right) \right|_{\theta =
L-B+\alpha}$. Next, we proceed with defining a set
$\tilde{\Lambda}$. The purpose of using this set is to transform the
summation indexes to corresponding $\left( \theta, \alpha \right)$
for each $\left( B, \alpha \right) \in \tilde{\Lambda}$. Now we show
that $\tilde{\Lambda}$ is a superset of $\Lambda$.This is done in
four steps.
\begin{enumerate}
\item First, recall that
\begin{equation}
\alpha \geq 0 .
\label{eq:app4-lemP1upperbound1}
\end{equation}
\item Second, observe that
\begin{eqnarray}
\left[ B \leq \frac{2 L}{3} - 1 \right] & \Longrightarrow &
\left[ L - \theta + \alpha \leq \frac{2 L}{3} - 1 \right] \nonumber \\
& \Longrightarrow & \left[ \alpha \leq \theta - \frac{L}{3} - 1 \right]
\label{eq:app4-lemP1upperbound2}
\end{eqnarray}
\item Third, note that \eqref{eq:app4-lemP1upperbound2} is equivalent to
\begin{equation}
\theta \geq \frac{L}{3} + \alpha + 1
\label{eq:app4-lemP1upperbound3}
\end{equation}
Using \eqref{eq:app4-lemP1upperbound1} in \eqref{eq:app4-lemP1upperbound3} implies
\begin{equation}
\theta \geq \frac{L}{3} + 1
\label{eq:app4-lemP1upperbound4}
\end{equation}
\item Fourth, using $B \leq \frac{2 L}{3} - 1$ in $\alpha \leq 2 B  - L$
(cf. \eqref{eq:app4-lemP1upperboundLambda}) implies
\begin{equation}
\alpha \leq \frac{L}{3} - 2 .
\label{eq:app4-lemP1upperbound5}
\end{equation}
Also, using \eqref{eq:app4-B} we have
\begin{equation}
\left[ \alpha \leq 2 B - L = L - 2 \theta + 2 \alpha \right]
\quad \Longrightarrow \quad
\left[ \theta \leq \frac{L}{2} + \frac{\alpha}{2} \right] .
\label{eq:app4-lemP1upperbound6}
\end{equation}
Using \eqref{eq:app4-lemP1upperbound5} in \eqref{eq:app4-lemP1upperbound6} yields
\begin{equation}
\theta \leq \frac{2 L}{3} - 1.
\label{eq:app4-lemP1upperbound7}
\end{equation}
\end{enumerate}
Now, defining
\[
\tilde{\Lambda} \eqdef \left\{ \left( B, \alpha \right) \; : \;
\frac{L}{3} + 1 \leq \theta \leq \frac{2 L}{3} - 1 , \, 0 \leq
\alpha \leq \theta - \frac{L}{3} - 1, \; \mbox{where} \; B = L - \theta + \alpha  \right\} ,
\]
and using \eqref{eq:app4-lemP1upperbound1}, \eqref{eq:app4-lemP1upperbound2},
\eqref{eq:app4-lemP1upperbound4}, \eqref{eq:app4-lemP1upperbound7}, we conclude that
$\Lambda \subseteq \tilde{\Lambda}$, which implies \eqref{eq:app4-P1upperbound}.
\end{proof}

Next, we proceed with providing an upper bound on the right hand side of \eqref{eq:app4-P1upperbound},
which will be shown to be $\cO \left(L^{-1} \right)$, i.e., diminishing in $L$, the length of the
generator polynomial of the LFSR\footnote{This result, in turn, implies that
an optimal QuBaR attack which uses the solution to the Sorting Problem II for
$\theta > L /3$ has a negligible cumulative success probability, i.e., negligible
success probability.}.
In order to achieve this task, we heavily use the concept of ``typical set''
(cf. \eqref{eq:typicalsetdefinition}).
Note that, using \eqref{eq:app4-Gset-prob} and $H(Q)=2$,
\eqref{eq:typicalsetdefinition} can be shown to be equivalent to
\begin{equation}
A_{\epsilon}^{(\theta)} = \left\{ \bq_1^{\theta} \; : \; 1 -
\epsilon \leq \frac{\beta}{\theta} \leq 1 + \epsilon \right\}.
\label{eq:app4-typ-set-1}
\end{equation}
In the following lemma, we show that all guesses $\left\{ \cG \left( \theta , \alpha \right) \right\}$
included in the summation of the right hand side of \eqref{eq:app4-P1upperbound}
are necessarily ``atypical'' (i.e., belong to the complement of the corresponding typical set).

\begin{lemma}
For any $\theta \in \bbZ^+$, such that $\theta > L/3$, and for all
$\alpha \in \bbN$, such that $0 \leq \alpha \leq \theta -
\frac{L}{3}$, we have $\cG\left( \theta, \alpha \right)
\subseteq [A_{\epsilon}^{(\theta)}]^{(c)}$ for all $\epsilon \in \left( 0,
\frac{2}{\theta}\right)$, where $[A_{\epsilon}^{(\theta)}]^{(c)}$
denotes the complement of the typical set $A_{\epsilon}^{(\theta)}$.
\label{lem:app4-lem1}
\end{lemma}
\begin{proof}
First of all, note that (cf. \eqref{eq:app4-Gset-def}), we have
\begin{equation}
\left[ \bq_1^{\theta} \in \cG(\theta, \alpha)\right] \Rightarrow
\left[ \frac{\beta}{\theta} = \left( \frac{L}{\theta} - 2 \right) +
\frac{\alpha}{\theta}\right].
\label{eq:app4-lem1-1}
\end{equation}
Hence, for any $\bq_1^{\theta} \in \cG(\theta,\alpha)$ such that
$\theta > L/3$ and $0 \leq \alpha \leq \theta - \frac{L}{3}$ , we
have
\begin{eqnarray}
-\frac{1}{\theta} \log p\left( \bq_1^{\theta}\right) - H(Q) & = &
\frac{\theta + \beta}{\theta} - 2 , \label{eq:app4-lem1-1.5} \\
& = & \frac{L-\theta+\alpha}{\theta} - 2 = \frac{L+\alpha}{\theta} - 3 , \label{eq:app4-lem1-2} \\
& \leq & \frac{2 L}{3 \theta} - 2 , \label{eq:app4-lem1-3} \\
 & \leq & \frac{2L}{L+3} - 2 = -\frac{6}{L+3} < 0 , \label{eq:app4-lem1-4}
\end{eqnarray}
where
\eqref{eq:app4-lem1-1.5} follows from the fact that $p \left( \bq_1^\theta \right)
 = 2^{- \left( \theta + \beta \right)}$ and $H \left( Q \right) = 2$,
\eqref{eq:app4-lem1-2} follows using \eqref{eq:app4-lem1-1} in \eqref{eq:app4-lem1-1.5},
\eqref{eq:app4-lem1-3} follows since $\alpha \leq \theta - L/3$,
\eqref{eq:app4-lem1-4} follows since $\theta \geq \frac{L}{3} + 1$.
Note that \eqref{eq:app4-lem1-4} implies
\begin{equation}
\left| -\frac{1}{\theta} \log p\left( \bq_1^{\theta} \right) - H(Q)
\right| \geq \frac{6}{L+3}.
\label{eq:app4-lem1-6}
\end{equation}
Now, since $\theta \geq \frac{L}{3} + 1$ (equivalently $\frac{2}{\theta} \leq \frac{6}{L+3}$),
we have $\epsilon < \frac{6}{L+3}$ for all $\epsilon \in \left( 0,
\frac{2}{\theta}\right)$. Using this in \eqref{eq:app4-lem1-6}, the claim follows.
\end{proof}

Next, we provide an upper bound on the right hand side of \eqref{eq:app4-P1upperbound}
using Lemma~\ref{lem:app4-lem1}. For all $\epsilon_\theta \in \left( 0,
\frac{2}{\theta}\right)$, we have
\begin{eqnarray}
\sum_{\theta = \frac{L}{3}+1}^{\frac{2L}{3}-1} \sum_{\alpha=0}^{\theta - \frac{L}{3}-1} \Pr\left( \cG \left( \theta, \alpha\right) \right)
& \leq & \sum_{\theta = \frac{L}{3}+1}^{\frac{2L}{3}-1} \Pr\left(\left[ A_{\epsilon_{\theta}}^{\theta} \right]^{(c)}\right), \label{eq:app4-7} \\
& \leq & \sum_{\theta = \frac{L}{3}+1}^{\frac{2L}{3}-1} \epsilon_{\theta}, \label{eq:app4-8}\\
& \leq & \left( \frac{L}{3}-1 \right)\left( \max_{\frac{L}{3}+1 \leq
\theta \leq \frac{2L}{3}-1} \epsilon_{\theta}\right),
\label{eq:app4-9}
\end{eqnarray}
where \eqref{eq:app4-7} follows from Lemma~\ref{lem:app4-lem1} and the fact that,
for any given $\theta$, $\left\{ \cG \left( \theta , \alpha \right) \right\}$ are disjoint by construction,
\eqref{eq:app4-8} follows from \eqref{eq:typicalset-probability}.
Now, choosing $\epsilon_\theta = \frac{1}{\theta^2}$ for all $\theta$, and using \eqref{eq:app4-9}
in \eqref{eq:app4-P1upperbound}, we have
\begin{equation}
P_1 \leq \left( \frac{L}{3}-1 \right)\left( \max_{\frac{L}{3}+1 \leq
\theta \leq \frac{2L}{3}-1} \frac{1}{\theta^2} \right)
= \frac{L/3-1}{\left(L/3+1\right)^2} < \frac{3}{L}.
\label{eq:app4-succ-prob-2}
\end{equation}
Thus, for any $\delta_1>0$, there exists a sufficiently large $L$ (per assumption A4), where
\begin{equation}
P_1 < \delta_1.
\label{eq:app4-succ-prob-2.5}
\end{equation}

Note that, for the optimal strategy, which uses the ordering
mentioned in \eqref{eq:app4-sorting-ex}, \eqref{eq:app4-succ-prob-2}
and \eqref{eq:app4-succ-prob-2.5}
imply that the range of $\frac{L}{2} \leq B \leq \frac{2 L}{3} - 1$
is not sufficient to achieve every given
cumulative success probability strictly greater than $1/2$,
since $\delta_1$ can be made arbitrarily small.
Therefore, we necessarily need to include guesses with $B=2L/3$ in the optimal
structure to achieve a cumulative success probability strictly greater than $1/2$.

Next, we proceed with quantifying the contribution to the cumulative success probability
for the case of $B = 2 L /3$. In this case, for the optimal strategy, since $\theta = L - B + \alpha$
and $0 \leq \alpha \leq 2 B - L$ for a given value of $B$, the corresponding
$\left( \theta , \alpha \right)$ pairs are of the form $\left\{ \left( \frac{L}{3} + \alpha , \alpha \right)
\right\}_{0 \leq \alpha \leq L/3}$. Thus, for the case of $B = 2 L /3$, the total contribution
to the cumulative success probability is given by
\begin{equation}
\Pr\left( \cG(L/3,0)\right) + \sum_{\alpha=1}^{L/3}
\Pr\left(\cG(L/3+\alpha, \alpha)\right). \label{eq:app4-10}
\end{equation}
Note that, the right hand side of \eqref{eq:app4-10}
is ``atypical'' per Lemma~\ref{lem:app4-lem1}; accordingly, we will show that
the only significant contribution to the cumulative success probability is due to the
left hand side of \eqref{eq:app4-10} since it includes terms within the corresponding
typical set.

Next, we provide an upper bound on the right hand side of \eqref{eq:app4-10}.
Defining $P_2 \eqdef \sum_{\alpha=1}^{L/3} \Pr \left(
\cG(L/3+\alpha, \alpha) \right)$, for all $\epsilon_\theta \in \left( 0 , \frac{2}{\theta} \right)$,
we have
\begin{eqnarray}
P_2 & = & \sum_{\theta = \frac{L}{3}+1}^{2L/3} \Pr\left( \cG(\theta, \theta-L/3)\right),\label{eq:app4-10.5}\\
 & \leq & \sum_{\theta = \frac{L}{3}+1}^{2L/3} \Pr\left( \left[ A_{\epsilon_{\theta}}^{\theta}\right]^{c}\right), \label{eq:app4-11} \\
 & \leq & \left( \frac{L}{3}\right) \left( \max_{L/3+1 \leq \theta \leq 2L/3}\epsilon_{\theta}\right),\label{eq:app4-12}
\end{eqnarray}
where \eqref{eq:app4-10.5} follows from using $\theta = \left. \left( L - B + \alpha \right) \right|_{B =  2L /3}$,
\eqref{eq:app4-11} follows from  Lemma~\ref{lem:app4-lem1},
\eqref{eq:app4-12} follows using \eqref{eq:typicalset-probability}.
Choosing $\epsilon_\theta = \frac{1}{\theta^2}$ for all $\theta$ in \eqref{eq:app4-12}, we have
\begin{equation}
P_2 \leq \frac{L/3}{\left( L/3+1\right)^2} < \frac{3}{L}.
\label{eq:app4-succ-prob-3}
\end{equation}
Thus, for any $\delta_2>0$, there exists a sufficiently large $L$ (per assumption A4), where
\begin{equation}
P_2 < \delta_2.
\label{eq:app4-succ-prob-3.5}
\end{equation}
Since $\delta_1$ (resp. $\delta_2$) in \eqref{eq:app4-succ-prob-2.5} (resp.
\eqref{eq:app4-succ-prob-3.5})  can be made arbitrarily small,
we necessarily need to use guesses from the set $\cG \left( \frac{L}{3} , 0 \right)$
in order to achieve a cumulative success probability strictly greater than $1/2$.

Next, consider the case of $\left( \theta, \alpha \right) = \left(
\frac{L}{3}, 0 \right)$: Note that, for any $\bq_1^{L/3} \in
\cG\left(\frac{L}{3},0\right)$, we have
\begin{equation}
p\left( \bq_1^{L/3}\right) = 2^{-(2L/3)}. \label{eq:app4-12.5}
\end{equation}
Per \eqref{eq:app4-typ-set-1}, \eqref{eq:app4-12.5} implies that
$\cG \left( \frac{L}{3},0 \right) \subseteq A_{\epsilon}^{\left( L /
3 \right)}$ for any $\epsilon > 0$. Furthermore, after some straightforward algebraic manipulations,
it can be shown that, for $0 < \epsilon <
\frac{3}{L}$, we have $A_{\epsilon}^{\left( L / 3 \right)} \subseteq
\cG \left( \frac{L}{3},0 \right)$; therefore we have
\begin{equation}
\cG \left( \frac{L}{3},0 \right) = A_{\epsilon}^{\left( L / 3 \right)} \;  \mbox{for} \; 0
< \epsilon < \frac{3}{L}.
\label{eq:converse-fund}
\end{equation}
In fact, \eqref{eq:converse-fund} constitutes the fundamental crux
of the converse proof. This observation implies that, using
sufficiently many guesses from the set $\cG \left( \frac{L}{3} , 0
\right)$ is both necessary (since $\delta_1$ and $\delta_2$ may be
arbitrarily small) and sufficient (since for $0 < \epsilon<
\frac{3}{L}$, we have $\Pr\left( \cG\left( \frac{L}{3},0 \right)
\right) = \Pr\left( A_\epsilon^{(L/3)}\right) > 1 - \epsilon $) to
achieve a cumulative success probability strictly greater than $1/2$
for large $L$ (per Assumption A4).

Now, let
\begin{equation}
P_1 + P_2 + P_3 > 1/2,
\label{eq:app4-cumulativetotal}
\end{equation}
denote the cumulative
success probability of optimal attack in the set $\tilde{\cS}_p^E$,
where $P_3$ denotes the contribution to the cumulative success probability
by the guesses from $\cG\left( \frac{L}{3}, 0\right)$\footnote{Note
that, w.l.o.g. we assume that, at step $B = 2 L / 3$, the proposed optimal attack
uses guesses from the
set $\cG\left( \frac{L}{3}, 0 \right)$ {\em in the end} (i.e., after applying guesses
from the sets $\left\{ \cG \left( \frac{L}{3} + \alpha , \alpha \right) \right\}_{\alpha=1}^{L/3}$
of which contributions to the cumulative success probability is denoted by $P_2$).
Since our strategy is to ``lower-bound'' the number of
guesses from the set $\cG \left( \frac{L}{3} , 0 \right)$ and declare the resulting value as a
lower bound on the overall complexity, $\tilde{\cC}_{min}$, this approach maintains the
validity of our result.}.
Using \eqref{eq:app4-succ-prob-2} and
\eqref{eq:app4-succ-prob-3} in \eqref{eq:app4-cumulativetotal}, we have
\begin{equation}
P_3 > \frac{1}{2} - \frac{6}{L}. \label{eq:app4-13}
\end{equation}

Next, let $\cC'$ denote the number of sequences used from the set
$\cG\left( \frac{L}{3}, 0\right)$. Using \eqref{eq:app4-12.5},
we have
\begin{equation}
\cC' = P_3 / 2^{-2L/3}\label{eq:app4-14}.
\end{equation}

Combining \eqref{eq:app4-13} and \eqref{eq:app4-14} yields
\begin{equation*}
\left[ \cC' > 2^{2L/3}\left( \frac{1}{2} - \frac{6}{L}\right)
\right] \quad \Longrightarrow \left[ \cC\left(
\tilde{\fA}_{opt}^E\right) > 2^{2L/3}\left( \frac{1}{2} -
\frac{6}{L}\right) \right],
\end{equation*}
since $\cC\left( \tilde{\fA}_{opt}^E\right) > \cC'$. Next, using
$\cS^E_p \subseteq \tilde{\cS}^E_p$ yields
\begin{equation}
\cC_{min}^E = \cC\left( \fA_{opt}^E \right) \geq  \tcC_{min}^E =
\cC\left( \tilde{\fA}_{opt}^E \right) > 2^{2L/3}\left( \frac{1}{2} -
\frac{6}{L}\right), \label{eq:app4-bound-2}
\end{equation}
where $\fA_{opt}^E$ and $\tilde{\fA}_{opt}^E$ have been defined in
\eqref{eq:app4-1} and \eqref{eq:app4-2}, respectively. Hence, the
claim finally follows. \qed

\section{Proof of Theorem~\ref{thrm:optimality-exhaustive}}
\label{app-5}
\setcounter{equation}{0}
\renewcommand{\theequation}{V-\arabic{equation}}

For the sake of clarity, we use the notation $G_k\left( i_k = 1,
\theta_k, \left(\bq_{i_k}^{\theta_k+i_k-1}\right)_k \right)$
(instead of $G_k \left( i_k = 1, \theta_k,  \bq_{i_k}^{\theta_k+i_k-1}  \right)$)
throughout the proof in this section.
\begin{itemize}
\item[(i)] First of all, note that  letting $\fA_{opt}^E =
\{G_k\}_{k=1}^{\mathcal{C}\left( \fA_{opt}^E \right)}$ denote the
optimal exhaustive-search type QuBaR attack against ABSG with
success probability $\Pr_{succ}\left( \fA_{opt}^E \right)$, the claim
is equivalent to the following statement: For any $i \neq j; i,j \in
\{ 1,\ldots, \mathcal{C}\left( \fA_{opt}^E \right) \}$ (assuming
$\theta_j > \theta_i$ w.l.o.g.), we have $\left( \bq_{1}^{\theta_i}
\right)_i \neq \left( \bq_{1}^{\theta_i} \right)_j$. Suppose to the
contrary, we have $\left( \bq_{1}^{\theta_i} \right)_i = \left(
\bq_{1}^{\theta_i} \right)_j$ for some $i \neq j; i,j \in \{
1,\ldots, \mathcal{C}\left( \fA_{opt}^E \right) \}$ where w.l.o.g.
$\theta_j > \theta_i$. Given $\fA_{opt}^E$, we construct an
exhaustive-search type QuBaR attack $\tilde{\fA}^E$ via eliminating
$G_j$ from $\fA_{opt}^E$, i.e., $\tilde{\fA}^E \eqdef \{
\tilde{G}_k\}_{k=1}^{\mathcal{C}\left( \fA_{opt}^E \right)-1}$ where
$\tilde{G}_k = G_k$ for $k \in \{ 1, \ldots, j-1\}$ and $\tilde{G}_k
= G_{k+1}$ for $k \in \{  j, \ldots, \mathcal{C}\left( \fA_{opt}^E
\right)-1\}$. Next, note that
\begin{equation}
\left[ \left( \cT\left( G_j \right)= 1 \right)  \; \Rightarrow \;
\left( \cT\left(G_i \right) = 1 \right) \right] \quad
\Longrightarrow \quad \left[ \Pr \left( \vee_{k=1}^{\mathcal C\left(
\fA_{opt}^E \right)} \left[ \cT\left(G_k \right) = 1 \right] \right) = \Pr
\left( \vee_{1 \leq k \leq {\mathcal C}\left( \fA_{opt}^E \right), \, k \neq
j} \left[ \cT\left(G_k \right) = 1 \right]
 \right) \right] .
\label{eq:practicalcode-lem1}
\end{equation}
If $G_j$ is a correct guess, then all $\left( \bq_1^{\theta_j}
\right)_j$ are correct, which implies $\left( \bq_1^{\theta_i}
\right)_j$ are necessarily correct as well since $\theta_i <
\theta_j$. Further, this implies that $\left( \bq_1^{\theta_i}
\right)_i$ are correct as well per the contradiction assumption.
Hence, this proves the left hand side of
(\ref{eq:practicalcode-lem1}); thus, the right hand side of
(\ref{eq:practicalcode-lem1}) is true as well. This, in turn, is
equivalent to $\Pr_{succ}\left(\fA_{opt}^E\right) =
\Pr_{succ}(\tilde{\fA}^E)$ which yields the promised contradiction ($\cC \left( \tilde{\fA}^E \right) = \cC \left( \fA_{opt}^E \right) - 1$)
since $\fA_{opt}^E$ is an optimal exhaustive-search type QuBaR attack
for the given success probability
$\Pr_{succ}\left(\fA_{opt}^E\right)$; hence the proof the first
statement of Theorem~\ref{thrm:optimality-exhaustive}.

\item[(ii)] Suppose not; then this means that there exists some $i,j \in \left\{
1 , 2 , \ldots , \mathcal{C}\left( \fA_{opt}^E \right) \right\}$ , $i \neq
j$, such that $\Pr \left[ \left( \cT\left( G_i \right) =1 \right) \cap \right.$ $\left. \left(
\cT\left( G_j \right) =1 \right) \right] > 0$. This implies that
there is some realization $\tilde{\bq}$ of $\bQ$ with non-zero
probability such that the events of $\left( \cT\left( G_i \right) =
1 \right)$ and $\left( \cT\left( G_j \right) = 1 \right)$ both occur
at the same time. In other words, there exists some $\tilde{\bq}$
with $\Pr \left( \bQ = \tilde{\bq} \right) > 0$ such that $\left(
\bq_1^{\theta_i} \right)_i = \tilde{\bq}_1^{\theta_i}$ and $\left(
\bq_1^{\theta_j} \right)_j = \tilde{\bq}_1^{\theta_j}$. However,
this implies that $\left( \bq_1^{\theta_i} \right)_i$ is a prefix of
$\left( \bq_1^{\theta_j} \right)_j$  (assuming w.l.o.g. $\theta_i <
\theta_j$). Hence contradiction (per the first statement of
Theorem~\ref{thrm:optimality-exhaustive}) and the proof of the
second statement of Theorem~\ref{thrm:optimality-exhaustive}.

\item[(iii)] This statement is the direct consequence of the first
and second statements of the theorem.

\item[(iv)] First recall that, at optimality $\cC\left( \fA_{opt}^E \right)$ is the smallest possible value (given the success
probability $\mbox{Pr}_{succ}\left(\fA_{opt}^E \right)$). This
observation and (\ref{eq:practicalcode-unionboundoptimality})
clearly imply that the optimal strategy consists of ``sorted''
guesses (in descending order) with respect to the probabilities
$\left\{ \Pr \left( \cT\left( G_k \right) =1 \right) \right\}$ of
the corresponding success events $\left\{ \left(\cT\left( G_k
\right) =1 \right) \right\}$ since the success probability
$\mbox{Pr}_{succ}\left(\fA_{opt}^E\right)$ is fixed.

\end{itemize}
\qed

\section{Proof of Theorem~\ref{thrm:achievability-general}}
\label{app-6}
\setcounter{equation}{0}
\renewcommand{\theequation}{VI-\arabic{equation}}

The attack mentioned in the statement of the theorem is the ``most
probable case attack'' given in \cite{ABSG, gou-sib:06}, which consists of simply
``trying'' a guess of the all zero sequence of $\left\{ Q_i \right\}$ (of length $L/2$)
for non-overlapping windows of output; here we assume w.l.o.g. that $L$ is sufficiently large and even
per assumption A4 of Sec.~\ref{ssec:assumptions}. Formally, this attack can be defined as follows:
\[
\mathfrak{A}_{ach,opt} = \{G_k\}_{k=1}^{\mathcal{C}\left(
\fA_{ach,opt} \right)}, \textrm{ s.t. for each guess $G_k = \left( i_k , \theta_k , \bq_{i_k}^{i_k+\theta_k-1} \right)$,}  \; i_k =
(k-1) \frac{L}{2} + 1 , \, \theta_k = \frac{L}{2} , \,
\beta_k = 0 ,
\]
where we recall that, for each $k$, $\beta_k = \sum_{j=0}^{\theta_k-1} q_{i_k+j}$ and the probability of $G_k$'s being
correct is $\Pr \left[ \cT \left( G_k \right) = 1 \right] = 2^{-\left( \theta_k + \beta_k \right)}$.
Since $\beta_k = 0$  and $\theta_k = \frac{L}{2}$ for each guess $G_k$ of the proposed attack, we have,
\begin{equation}
\Pr \left[ \cT \left( G_k \right) = 1 \right] = 2^{-L/2} , \quad 1 \leq k \leq \cC \left( \fA_{ach,opt} \right) .
\label{eq:app6-1}
\end{equation}
Hence, we have
\begin{eqnarray}
\mbox{Pr}_{succ}(\mathfrak{A}_{ach,opt})
& = & \Pr\left( \vee_{k=1}^{\mathcal{C}\left( \fA_{ach,opt} \right)} \left[ \cT\left( G_k \right) = 1 \right]\right), \nonumber \\
& = & 1 - \Pr\left( \wedge_{k=1}^{\mathcal{C}\left( \fA_{ach,opt} \right)} \left[ \cT\left( G_k \right) = 0 \right]\right), \nonumber \\
 & = & 1 - \prod_{k=1}^{\mathcal{C}\left( \fA_{ach,opt} \right)}\Pr\left[ \cT\left( G_k \right) = 0 \right], \label{eq:thrm-achievability-general-1} \\
 & = & 1 - \left( 1 - 2^{-L/2} \right)^{\mathcal{C}\left( \fA_{ach,opt} \right)}, \label{eq:thrm-achievability-general-2}
\end{eqnarray}
where \eqref{eq:thrm-achievability-general-1} follows from the fact
that the events of $\left\{ \cT \left( G_k \right) = 0 \right\}$ are independent (since they correspond to sequences
of non-overlapping windows of $\left\{ Q_i \right\}$, which are i.i.d.)
and \eqref{eq:thrm-achievability-general-2} follows from \eqref{eq:app6-1}. Now, recall that
\begin{equation}
\lim_{x \rightarrow 0} \left( 1 - x \right)^{1/x} = 1/e
\label{eq:app6-2}
\end{equation}
Next, choosing $\mathcal{C}\left( \fA_{ach,opt} \right)=2^{L/2}$, for large $L$ we have
\[
\lim_{L \rightarrow \infty} \mbox{Pr}_{succ}(\mathfrak{A}_{ach,opt})
 = \lim_{L \rightarrow \infty} \left[ 1 - \left( 1 - 2^{-L/2} \right)^{ \left( 2^{L/2} \right)} \right] = 1 - \frac{1}{e} > \frac{1}{2} ,
\]
which follows from \eqref{eq:app6-2}. This implies that $\fA_{ach,opt} \in \cS_p$, where $\cC \left( \fA_{ach,opt} \right) = 2^{L/2}$
for sufficiently large $L$ (per assumption A4). Furthermore, note that all guesses $\left\{ G_k \right\}$ of the
proposed attack $\fA_{ach,opt}$ are equally-likely to succeed (cf. \eqref{eq:app6-1}), which subsequently implies that
$\cC_{ave} \left( \fA_{ach,opt} \right) = \frac{1}{2} \left( 2^{L/2} + 1 \right)$. Hence the proof.
\qed

\section{Proof of Theorem~\ref{thrm:converse-general}}
\label{app-7}
\setcounter{equation}{0}
\renewcommand{\theequation}{VII-\arabic{equation}}

Throughout the proof, we assume w.l.o.g. $L$ is even, since it is
sufficiently large per assumption A4. We first recall that we have
\begin{equation}
\forall \, k \in \bbZ^+ , \; \quad \Pr\left( \cT\left(G_k\right) = 1 \right)
\leq \left(\frac{1}{2}\right)^{L/2},
\label{eq:thrm-converse-general-0.5}
\end{equation}
due to \eqref{eq:upperboundonsuccess} of Appendix~\ref{app-4}
\footnote{Note that, in Appendix~\ref{app-4}, we derived
\eqref{eq:upperboundonsuccess} for exhaustive-search attacks, for which
the starting index of the attack is set to unity (cf. \eqref{eq:app4-Gset-def}).
However, after some straightforward algebra, it can be shown that,
following \eqref{eq:app4-Gset-def},
all the subsequent derivations of  Appendix~\ref{app-4}, regarding the ``valid ranges of
fundamental system parameters'', $\theta$, $\beta$, $\alpha$ and $B$,
(including the utilized result \eqref{eq:upperboundonsuccess}) are still valid even if
we relax the aforementioned condition on the starting index, which amounts to the
general case attacks. Thus, \eqref{eq:upperboundonsuccess} can be shown to
hold in the case of general QuBaR attacks.}.

Next, we proceed with a similar approach to the one pursued in the proof of
Theorem~\ref{thrm:converse-exhaustive}. In particular, we begin with
defining a set  $\tcS_p$, which is a superset of $\cS_p$,
the set of successful QuBaR attacks (cf. \eqref{eq:attack-class}):
\begin{equation}
\tilde{\cS}_p \eqdef
\left\{\mathfrak{A}=\{G_k\}_{k=1}^{\mathcal{C}\left( \fA \right)} :
\sum_{k=1}^{\mathcal{C}\left( \fA \right)}\Pr[\cT\left(G_k
\right)=1]>1/2\right\}. \label{eq:app7-1}
\end{equation}
Using the union bound yields
\[
\Pr\left(\vee_{k=1}^{\mathcal{C}\left( \fA \right)}\left[\cT\left(
G_k \right)=1 \right]\right) \leq \sum_{k=1}^{\mathcal{C}\left( \fA
\right)}\Pr\left(\cT\left( G_k \right)=1\right) .
\]
Thus, we have
\begin{equation}
\left[ \fA \in \cS_p \right] \quad \Longrightarrow \quad
\left[ \fA \in \tcS_p \right]  ,
\nonumber
\end{equation}
which implies
\begin{equation}
\cS_p \subseteq \tcS_p.
\label{eq:app7-2}
\end{equation}
Further, for any
$\mathfrak{A} \in \tilde{\mathcal{S}}_p$, we have
\begin{equation}
\frac{1}{2} < \sum_{k=1}^{\mathcal{C}\left( \fA \right)} \Pr\left[ \cT\left(G_k
\right) = 1 \right] \leq \sum_{k=1}^{\mathcal{C}\left( \fA \right)}
(1/2)^{L/2} = 2^{-L/2} \cC \left( \fA \right), \label{eq:thrm-converse-general-1}
\end{equation}
where the first and the second inequalities follow from \eqref{eq:app7-1} and
\eqref{eq:thrm-converse-general-0.5}, respectively.
As a result, we have
\[
\min_{\fA \in \cS_p} \cC \left( \fA \right) \geq \min_{\fA \in \tcS_p} \cC \left( \fA \right)
> 2^{L/2-1} ,
\]
where the first inequality follows from \eqref{eq:app7-2} and the second inequality
follows from the fact that \eqref{eq:thrm-converse-general-1}
holds for any $\fA \in \tcS_p$. Hence the proof.
\qed

\section*{Acknowledgement}
Authors wish to thank Nafiz Polat Ayerden and Mustafa Orhan Dirik of
Bo\u{g}azi\c{c}i University, Turkey for various helpful discussions
and comments.

\end{document}